\pdfoutput=1
\RequirePackage{ifpdf}
\ifpdf 
\documentclass[pdftex]{sigma}
\else
\documentclass{sigma}
\fi

\numberwithin{equation}{section}

\newtheorem{Theorem}{Theorem}[section]
\newtheorem{Corollary}[Theorem]{Corollary}
\newtheorem{Lemma}[Theorem]{Lemma}
\newtheorem{Proposition}[Theorem]{Proposition}
 { \theoremstyle{definition}

 }

\usepackage{mathrsfs}
\usepackage{euscript}
\usepackage{mathtools} 

\usepackage{silence} 
\usepackage{tikz} 
	\usetikzlibrary{arrows} 
	\usetikzlibrary{decorations.markings} 
	\usetikzlibrary{patterns} 
	\usetikzlibrary{calc} 

\usepackage[caption=false, font=small]{subfig}

\def\Z{\mathbb{Z}} 
\def\C{\mathbb{C}} 
\def\Im{\textup{Im}}

\begin{document}

\allowdisplaybreaks

\newcommand{\arXivNumber}{2004.09924}

\renewcommand{\thefootnote}{}

\renewcommand{\PaperNumber}{101}

\FirstPageHeading

\ShortArticleName{A Combinatorial Description of Certain Polynomials Related to the XYZ Spin Chain}

\ArticleName{A Combinatorial Description of Certain\\ Polynomials Related to the XYZ Spin Chain\footnote{This paper is a contribution to the Special Issue on Elliptic Integrable Systems, Special Functions and Quantum Field Theory. The full collection is available at \href{https://www.emis.de/journals/SIGMA/elliptic-integrable-systems.html}{https://www.emis.de/journals/SIGMA/elliptic-integrable-systems.html}}}

\Author{Linnea HIETALA}

\AuthorNameForHeading{L.~Hietala}

\Address{Department of Mathematics, Chalmers University of Technology\\ and University of Gothenburg, 412 96 Gothenburg, Sweden}
\Email{\href{mailto:linnea.hietala@gu.se}{linnea.hietala@gu.se}}

\ArticleDates{Received April 22, 2020, in final form September 24, 2020; Published online October 07, 2020}

\Abstract{We study the connection between the three-color model and the polynomials~$q_n(z)$ of Bazhanov and Mangazeev, which appear in the eigenvectors of the Hamiltonian of the XYZ spin chain. By specializing the parameters in the partition function of the 8VSOS model with DWBC and reflecting end, we find an explicit combinatorial expression for $q_n(z)$ in terms of the partition function of the three-color model with the same boundary conditions. Bazhanov and Mangazeev conjectured that $q_n(z)$ has positive integer coefficients. We prove the weaker statement that $q_n(z+1)$ and $(z+1)^{n(n+1)}q_n(1/(z+1))$ have positive integer coefficients. Furthermore, for the three-color model, we find some results on the number of states with a given number of faces of each color, and we compute strict bounds for the possible number of faces of each color.}

\Keywords{eight-vertex SOS model; domain wall boundary conditions; reflecting end; three-color model; partition function; XYZ spin chain; polynomials; positive coefficients}

\Classification{82B23; 05A15; 33E17}

\tikzset{midarrow/.style={
 decoration={markings,
 mark= at position 0.5 with {\arrow{#1}} ,
 },
 postaction={decorate}
		}
}

\renewcommand{\thefootnote}{\arabic{footnote}}
\setcounter{footnote}{0}

\section{Introduction}
The first example of a six-vertex (6V) model was introduced to describe ice. In this original ice-model, all vertex types, and thus all states, have the same weight. This and some other special cases of the 6V model were solved in 1967 by Lieb~\cite{Lieb1967}. The same year, Sutherland~\cite{Sutherland1967} solved the general 6V model. Lenard~\cite{Lieb1967} (note added in proof) found a bijection from the states of the 6V model to three-colorings of a square lattice such that no adjacent squares have the same color and with the color in one corner fixed. Baxter~\cite{Baxter1970} introduced the three-color model by assigning a weight to each color.

One of the first nontrivial examples of fixed boundaries were the domain wall boundary conditions (DWBC)~\cite{Korepin1982}. There are several ways to describe the 6V model with DWBC, for example with alternating sign matrices (ASMs) or height matrices (see, e.g.,~\cite{Propp2002}).
In 1996, Zeilberger~\cite{Zeilberger1996} proved the alternating sign matrix conjecture of Mills, Robbins and Rumsey~\cite{MillsRobbinsRumsey1983}, which gives a formula for the number of ASMs. Izergin~\cite{Izergin1987, IzerginCokerKorepin1992} showed that the partition function of the 6V model with DWBC can be expressed as a determinant, which Kuperberg~\cite{Kuperberg1996} used to give another proof of the alternating sign matrix conjecture.

The eight-vertex (8V) model is a generalization of the 6V model. To solve the 8V model, Baxter~\cite{Baxter1973} introduced the eight-vertex solid-on-solid (8VSOS) model, which is a two parameter generalization of the 6V model. The name is a bit misleading, since it has only six different local states, and therefore the 8VSOS model is also called the elliptic SOS model.

Tsuchiya \cite{Tsuchiya1998} obtained a determinant formula for the partition function of the 6V model with one reflecting end and DWBC on the three other sides. Kuperberg~\cite{Kuperberg2002} used this to enumerate the corresponding UASMs, which are alternating sign matrices, with U-turns on one side. The UASMs generalize the vertically symmetric alternating sign matrices (VSASMs). In~2011, Filali~\cite{Filali2011} found a single determinant formula for the partition function of the 8VSOS model with DWBC and one reflecting end. For the 8VSOS model with DWBC, but without the reflecting end, no simple determinant formula has been found.

\looseness=1 Razumov and Stroganov \cite{RazumovStroganov2001} found connections between the supersymmetric XXZ spin chain and ASMs. This has developed into a large area of research, see, e.g.,~\cite{Zinn-Justin2008}. Similar problems for the supersymmetric XYZ spin chain were studied by Bazhanov and Mangazeev.
In~\cite{BazhanovMangazeev2005}, they investigated the eigenvalues of Baxter's $Q$-operator \cite{Baxter1972} for the 8V model, and in~\cite{BazhanovMangazeev2010} (see also~\cite{RazumovStroganov2010}) they studied the Hamiltonian of the XYZ spin chain of odd length. The ground state eigenvalues of the $Q$-operator as well as the components of the ground state eigenvectors of the XYZ-Hamiltonian can be expressed in terms of certain polynomials. These polynomials seem to have positive integer coefficients \cite{BazhanovMangazeev2006, BazhanovMangazeev2010}, which suggests that the polynomials could have a~combinatorial interpretation. Up till now, no such interpretation has been presented.

\looseness=1 In \cite{Rosengren2011}, Rosengren extended Kuperberg's work from the 6V model to the 8VSOS model. Kuper\-berg's specialization of the parameters in the 6V model gives the ice model, and the same specialization in the 8VSOS model gives the three-color model. Again polynomials with positive coefficients showed up. Rosengren~\cite{Rosengren2015} introduced certain polynomials $T(x_1, \dots, x_{2n})$, which are generalizations of the polynomials in~\cite{Rosengren2011}. Zinn-Justin~\cite{Zinn-Justin2013} introduced polynomials equivalent to Rosengren's, and observed that Bazhanov's and Mangazeev's polynomials seem to be specializations of these polynomials. This indicates that the combinatorial interpretation of the polynomials with positive coefficients could be connected to three-colorings.

\looseness=1 In this paper, we study the link between the three-color model and the polynomials $q_n(z)$ of Bazhanov and Mangazeev, which appear in the eigenvectors of the XYZ-Hamiltonian~\cite{BazhanovMangazeev2010}. By specializing the parameters in the partition function of the 8VSOS model with DWBC and reflecting end in Kuperberg's way and then using Filali's determinant formula and Rosengren's polynomials $T(x_1, \dots, x_n)$, we can find an explicit combinatorial expression for $q_n(z)$ in terms of the partition function of the three-color model with the same boundary conditions.

\looseness=1 The outline of this paper is as follows. First of all, in Section~\ref{secprel}, we describe the 8VSOS model and the three-color model with DWBC and reflecting end, and in particular we define their partition functions.
Following Rosengren, in Section~\ref{sec3}, we specialize the parameters in the partition function of the 8VSOS model with DWBC and reflecting end to obtain the partition function of the three-color model. In Section~\ref{sec4}, we rewrite Filali's determinant formula to depend on the polynomials $q_{n-1}(z)$, going via Rosengren's polynomials $T(x_1, \dots, x_{2n})$. Then we compare the determinant formula with the expression from Section~\ref{sec3}, and get an expression for Bazhanov's and Mangazeev's polynomials $q_n(z)$ in terms of the three-colorings.

In this paper, we consider the three-color model on a square lattice with $(2n+1) \times (n+1)$ faces. The faces are filled with three different colors (which we call color $0$, $1$, and $2$), such that adjacent faces have different colors. We~consider the following boundary conditions. In~the upper left corner, we fix color $0$.
On three of the boundaries the colors alternate cyclically,
whereas on the left boundary of the lattice, each second face has color $0$. The remaining faces on the left boundary each has one of the other two colors (see Fig.~\ref{fig:3colormodel}). Details can be found in~Sec\-tion~\ref{3cmodel}.

In Section~\ref{secmainresult}, we simplify the expression from Section~\ref{sec4}. We~get the following theorem, which is our main result.

\begin{Theorem}\label{maintheorem}Let $t_i$ be the weight assigned to a face with color $i$, and let $m$ be the number of faces on the left boundary with color $2$. For a given $m$, it holds that
\begin{gather*}
 \sum_{(\textup{states with $m$ specified})}\prod_{\textup{faces}}t_i\\
\qquad =\begin{cases}
\displaystyle \binom{n}{m} \frac{t_2^{m-n}}{t_1^m}\frac{t_0(t_0t_1t_2)^{(2n^2+4n)/3}}{(z(z^2-1))^{(n^2-n)/3}}q_{n-1}(z)
, & \textup{$n\equiv 0, 1 \mod 3$},\\
\displaystyle \binom{n}{m} \frac{t_2^{m-n}}{t_1^m}\frac{(t_0t_1+t_0t_2+t_1t_2)(t_0t_1t_2)^{(2n^2+4n-1)/3}}{(3z^2+1)(z(z^2-1))^{(n^2-n-2)/3}}q_{n-1}(z),
& \textup{$n\equiv 2 \mod 3$},
\end{cases}
\end{gather*}
where $z$ is defined such that
\begin{gather*}
\frac{(t_0t_1+t_0t_2+t_1t_2)^3}{(t_0t_1t_2)^2}=\frac{(3z^2+1)^3}{(z(z^2-1))^2}.
\end{gather*}
\end{Theorem}

Consequences of Theorem~\ref{maintheorem} are discussed in Section~\ref{secconsequences}. The theorem yields an explicit expression for $q_n(z)$ in terms of three-colorings. Unfortunately it is not directly clear that the expression has positive coefficients, but we can prove the weaker result that $q_n(z+1)$ and $(z+1)^{n(n+1)}q_n(1/(z+1))$ have positive integer coefficients (Corollary~\ref{cor:poscoeffs}). Let $N^{(m)}(k_0, k_1, k_2)$ denote the number of states with exactly $m$ faces on the left boundary with color $2$, and $k_i$ entries of color~$i$. We~find that
\begin{gather*}
N^{(m)}(k_0, k_1, k_2)=\binom{n}{m}N^{(0)}(k_0, k_1+m, k_2-m)
\end{gather*}
(Corollary~\ref{numberofstateswith3colors}), and
we also find symmetries in the number of states with a given number of faces of each color (Corollary~\ref{symmetriesofcolors}). Furthermore we compute strict bounds for the possible number of faces of each color (Corollary~\ref{maxminfaces}).

\section{Preliminaries}\label{secprel}

Let $p={\rm e}^{2\pi{\rm i}\tau}$ and $q={\rm e}^{2\pi{\rm i} \eta}$, where $\tau$ and $\eta$ are fixed parameters with $\Im(\tau)>0$ and $\eta\notin \Z+\tau\Z$. By $q^x$ we will always mean ${\rm e}^{2\pi{\rm i} \eta x}$, and when we write $p^{1/2}$, we will mean $p^{1/2}={\rm e}^{\pi{\rm i} \tau}$. We~define the theta function
\begin{gather*}
\vartheta(x,p)= \prod_{j=0}^{\infty} \big(1-p^j x\big)\big(1-p^{j+1}/x\big).
\end{gather*}
Then we define $[x]=q^{-x/2}\vartheta(q^x, p)$. Sometimes we will write $\vartheta(x^{\pm a},p):=\vartheta(x^a, p)\vartheta(x^{-a}, p)$, or~we will suppress the $p$ and write $\vartheta(x):=\vartheta(x, p)$, and write out the second parameter only when it is not just~$p$.

Observe that $\vartheta(1)=0$. The most important properties of the theta function are
\begin{gather*}
\vartheta(px)=\vartheta(1/x)=-\frac{1}{x}\vartheta(x)
\end{gather*}
and the addition rule
\begin{gather}
\vartheta(x_1x_3)\vartheta(x_1/x_3)\vartheta(x_2x_4)\vartheta(x_2/x_4) -\vartheta(x_1x_4)\vartheta(x_1/x_4)\vartheta(x_2x_3)\vartheta(x_2/x_3)\nonumber\\
\qquad{} =\frac{x_2}{x_3}\vartheta(x_1x_2)\vartheta(x_1/x_2)\vartheta(x_3x_4)\vartheta(x_3/x_4).\label{additionrule}
\end{gather}

\subsection{The 8VSOS model with DWBC and reflecting end}

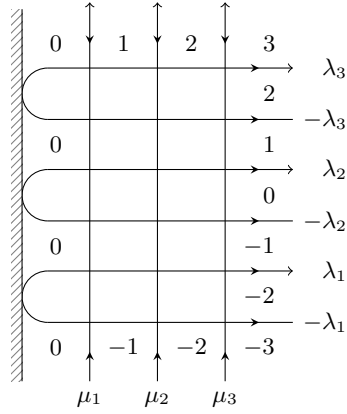
\begin{figure}[t]
\vspace{3mm}
\centering
\begin{tikzpicture}[scale=0.9,font=\footnotesize]
	\foreach \y in {1,...,3} {
		\draw (.38,1.5*\y-.25-.38) -- +(2.3+0.32,0);
		\draw[midarrow={stealth}] (3,1.5*\y-.25-.38) -- +(1,0) node[right]{$-\lambda_{\y}$};
		\draw (.38,1.5*\y-.25+.38) -- +(2.3+0.32,0);
		\draw[midarrow={stealth},->] (3,1.5*\y-.25+.38) -- +(1,0) node[right]{$\phantom{-}\lambda_{\y}$};
		\draw (0.38,1.5*\y-.25+.38) arc (90:270:0.38);		
	}
	
	\foreach \x in {1,...,3} {
		\draw[midarrow={stealth}] (\x,0) node[below]{$\mu_{\x}$} -- +(0,.87); 
		\draw (\x,.87) -- +(0,3.76); 
		
		\draw[midarrow={stealth reversed}, ->] (\x,4.63) -- +(0,.97);	
	}
	
	\fill[preaction={fill,white},pattern=north east lines, pattern color=gray] (0,0) rectangle (-.15,5.5) ; \draw (0,0) -- (0,5.5);
	
		\node at (0.5, 5) {$0$};
 \node at (1.5, 5) {$1$};
 \node at (2.5, 5) {$2$};
 \node at (3.65, 5) {$3$};

 \node at (3.65, 4.25) {$2$};
 \node at (3.65, 3.5) {$1$};
 \node at (3.65, 2.75) {$0$};
 \node at (3.5, 2) {$-1$};
 \node at (3.5, 1.25) {$-2$};

 \node at (3.5, 0.5) {$-3$};
 \node at (2.5, 0.5) {$-2$};
 \node at (1.5, 0.5) {$-1$};
 \node at (0.5, 0.5) {$0$};
		
		\node at (0.5, 2) {$0$};
		\node at (0.5, 3.5) {$0$};
\end{tikzpicture}
\caption{The 8VSOS model with DWBC and reflecting end in the case $n=3$. The parameters~$\mu_i$ and~$\lambda_i$ are the spectral parameters.}\label{fig:arrowpicmodel}
\end{figure}

Consider a $2n\times n$ square lattice, where the horizontal lines are connected pairwise at the left edge. Each such pair of horizontal lines can be thought of as one single line turning at a wall on the left side, see Fig.~\ref{fig:arrowpicmodel}.
Define $V$ as a two-dimensional complex vector space with basis vectors~$e_+$ and~$e_-$. To each line we associate a copy of~$V$, and we assign a spin $\pm 1$ to each edge. A~lattice with a spin assigned to each edge is called a state.

Graphically a state can be represented by giving each line a positive direction, which goes upwards for the vertical lines, to the left for the lower part of the horizontal double line, and to the right for the upper part. The positive direction is indicated by an arrow at the end of a~line.
Spin~$+1$ corresponds to an arrow pointing in the positive direction of the line, and spin~$-1$ corresponds to an arrow pointing in the opposite direction. This graphical notation follows~\cite{Lamers2016}.

At each vertex, the spins of the four surrounding edges need to obey the ice rule, that is, at~each vertex with spins $\alpha$, $\beta$, $\alpha^\prime$ and $\beta^\prime$ as in Fig.~\ref{fig:vertex}, the equation
\begin{gather*}
\alpha+\beta=\alpha^\prime+\beta^\prime
\end{gather*}
must hold. This yields six possible types of vertices, see Fig.~\ref{fig:6vmodel}. Because of the reflecting end, for every second row in the square lattice, we need to rotate the possible vertices (in Fig.~\ref{fig:6vmodel}) 90~degrees counterclockwise.

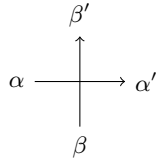
\begin{figure}[tb]
\centering
	\begin{tikzpicture}[baseline={([yshift=-.5*10pt*0.6]current bounding box.center)}, scale=0.6, font=\footnotesize]
		\draw[->] (0,1) node[left]{$\alpha$} -- (2,1) node[right]{$\alpha^\prime$};
		\draw[->] (1,0) node[below]{$\beta$} -- (1,2) node[above]{$\beta^\prime$};
	\end{tikzpicture}
\caption{A vertex with spins $\alpha$, $\beta$, $\alpha^\prime$, $\beta^\prime$ on the surrounding edges.}
\label{fig:vertex}
\end{figure}

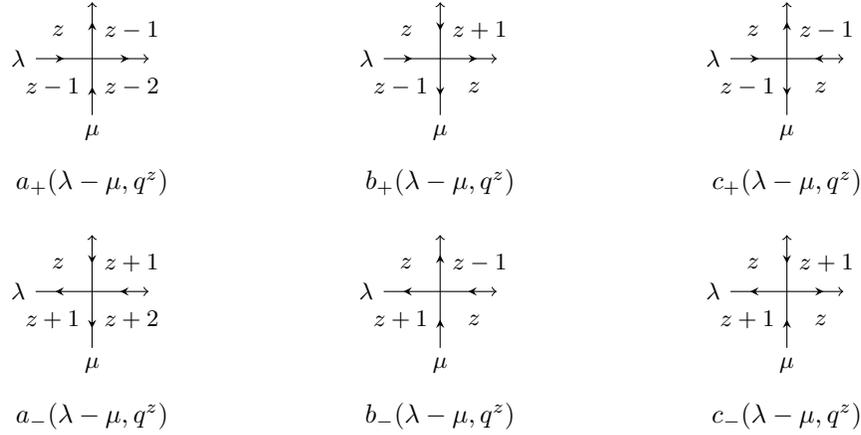
\begin{figure}[t]
\centering
 \subfloat{%
 	\begin{tikzpicture}[scale=0.75, font=\footnotesize]
 	\draw[midarrow={stealth}] (0,1) node[left] {$\lambda$} -- (1,1); 
		\draw[midarrow={stealth reversed}, <-] (2,1) -- (1,1); 
		\draw[midarrow={stealth}] (1,0) node[below] {$\mu$} -- (1,1); 
		\draw[midarrow={stealth reversed}, <-] (1,2) -- (1,1); 
		\draw (1,-1.2) node{\small{$
		a_+(\lambda-\mu, q^z)$}};
		\node at (0.4, 1.5) {$z$};
		\node at (1.7, 1.5) {$z-1$};
		\node at (0.3, 0.5) {$z-1$};
		\node at (1.7, 0.5) {$z-2$};
		\end{tikzpicture}
	}\hfil
	\subfloat{%
		\begin{tikzpicture}[scale=0.75, font=\footnotesize]
		\draw[midarrow={stealth}] (0,1) node[left] {$\lambda$} -- (1,1); 
		\draw[midarrow={stealth reversed}, <-] (2,1) -- (1,1); 
		\draw[midarrow={stealth reversed}] (1,0) node[below] {$\mu$} -- (1,1); 
		\draw[midarrow={stealth}, <-] (1,2) -- (1,1); 
		\draw (1,-1.2) node{\small{$
		b_+(\lambda-\mu, q^z)$}};
		\node at (0.4, 1.5) {$z$};
		\node at (1.7, 1.5) {$z+1$};
		\node at (0.3, 0.5) {$z-1$};
		\node at (1.6, 0.5) {$z$};
		\end{tikzpicture}
	}\hfil
	\subfloat{%
		\begin{tikzpicture}[scale=0.75, font=\footnotesize]
		\draw[midarrow={stealth}] (0,1) node[left] {$\lambda$} -- (1,1); 
		\draw[midarrow={stealth}, <-] (2,1) -- (1,1); 
		\draw[midarrow={stealth reversed}] (1,0) node[below] {$\mu$} -- (1,1); 
		\draw[midarrow={stealth reversed}, <-] (1,2) -- (1,1); 
		\draw (1,-1.2) node{\small{$
		c_+(\lambda-\mu, q^z)$}};
		\node at (0.4, 1.5) {$z$};
		\node at (1.7, 1.5) {$z-1$};
		\node at (0.3, 0.5) {$z-1$};
		\node at (1.6, 0.5) {$z$};
		\end{tikzpicture}
	}\\
	\vspace{0mm}
 \subfloat{%
		\begin{tikzpicture}[scale=0.75, font=\footnotesize]
		\draw[midarrow={stealth reversed}] (0,1) node[left] {$\lambda$} -- (1,1); 
		\draw[midarrow={stealth}, <-] (2,1) -- (1,1); 
		\draw[midarrow={stealth reversed}] (1,0) node[below] {$\mu$} -- (1,1); 
		\draw[midarrow={stealth}, <-] (1,2) -- (1,1); 
		\draw (1,-1.2) node{\small{$
		a_-(\lambda-\mu, q^z)$}};
		\node at (0.4, 1.5) {$z$};
		\node at (1.7, 1.5) {$z+1$};
		\node at (0.3, 0.5) {$z+1$};
		\node at (1.7, 0.5) {$z+2$};
		\end{tikzpicture}
	}\hfil
	\subfloat{%
		\begin{tikzpicture}[scale=0.75, font=\footnotesize]
		\draw[midarrow={stealth reversed}] (0,1) node[left] {$\lambda$} -- (1,1); 
		\draw[midarrow={stealth}, <-] (2,1) -- (1,1); 
		\draw[midarrow={stealth}] (1,0) node[below] {$\mu$} -- (1,1); 
		\draw[midarrow={stealth reversed}, <-] (1,2) -- (1,1); 
		\draw (1,-1.2) node{\small{$
		b_-(\lambda-\mu, q^z)$}};
		\node at (0.4, 1.5) {$z$};
		\node at (1.7, 1.5) {$z-1$};
		\node at (0.3, 0.5) {$z+1$};
		\node at (1.6, 0.5) {$z$};
		\end{tikzpicture}
	}\hfil
	\subfloat{%
		\begin{tikzpicture}[scale=0.75, font=\footnotesize]
		\draw[midarrow={stealth reversed}] (0,1) node[left] {$\lambda$} -- (1,1); 
		\draw[midarrow={stealth reversed}, <-] (2,1) -- (1,1); 
		\draw[midarrow={stealth}] (1,0) node[below] {$\mu$} -- (1,1); 
		\draw[midarrow={stealth}, <-] (1,2) -- (1,1); 
		\draw (1,-1.2) node{\small{$
		c_-(\lambda-\mu, q^z)$}};
		\node at (0.4, 1.5) {$z$};
		\node at (1.7, 1.5) {$z+1$};
		\node at (0.3, 0.5) {$z+1$};
		\node at (1.6, 0.5) {$z$};
		\end{tikzpicture}
	}\\
	\vspace{-1mm}
	\caption{The possible vertex weights for the 8VSOS model. The spins are indicated with an arrow halfway the edge, where right and up are positive spins, and left and down are negative spins. The vertex weights also depend on the spectral parameters $\lambda$ and $\mu$, as well as the height $z$ in the upper left face.}
	\label{fig:6vmodel}
	\label{fig:arrowpic}
\end{figure}

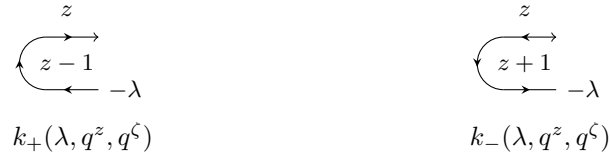
\begin{figure}[t]
\centering
\subfloat{%
 \begin{tikzpicture}[baseline={([yshift=-.5*10pt*0.6]current bounding box.center)}, scale=0.7, font=\footnotesize]
		\draw (0.5,0.5) arc (90:270:0.5);
		\draw[midarrow={stealth reversed}] (.5,-.5) -- (1.5,-.5) node[right] {$-\lambda$};
		\draw[midarrow={stealth}, ->] (.5,+.5) -- (1.5,+.5);
		\draw[-stealth] (0,0.05) -- (0,0.06);
		
	\node at (0.9, 1) {$z$};
	\node at (0.9, 0) {$z-1$};
	
	\draw (1.2,-1.4) node{\small{$
	k_+(\lambda, q^z, q^\zeta)$}};
 \end{tikzpicture}
}\hfil
\subfloat{%
 \begin{tikzpicture}[baseline={([yshift=-.5*10pt*0.6]current bounding box.center)}, scale=0.7, font=\footnotesize]
		\draw (0.5,0.5) arc (90:270:0.5);
		\draw[midarrow={stealth}] (.5,-.5) -- (1.5,-.5) node[right] {$-\lambda$};
		\draw[midarrow={stealth reversed},->] (.5,+.5) -- (1.5,+.5);
		\draw[-stealth] (0,-0.05) -- (0,-0.06);
		
	\node at (0.9, 1) {$z$};
	\node at (0.9, 0) {$z+1$};
	
	\draw (1.2,-1.4) node{\small{$	k_-(\lambda, q^z, q^\zeta)$}};
 \end{tikzpicture}
}
\vspace{-1mm}
\caption{The possible boundary weights for the reflecting ends that we consider in this model. The weights depend on the spectral parameter $\lambda$ and the height $z$ outside the turn, as well as on a boundary parameter $\zeta$.}
 \label{fig:reflectingends}
		\label{fig:arrowpicreflend}
\end{figure}

Fix a dynamical parameter $\rho\in\C$. To each face we assign a height $\rho+a$, $a\in\Z$. Heights of adjacent faces should always differ by $\pm 1$. Given a face with height $z$, crossing an edge of spin $s=\pm 1$ from the left to the right (looking in the positive direction of the edge) yields the height $z-s$ in the adjacent face (see Fig.~\ref{fig:6vmodel}). Hence, for a given state, it is enough to specify the height in one place, which we choose to be in the upper left corner. Defining $\rho$ to be the height in the upper left corner, we can write the heights minus $\rho$, as in Fig.~\ref{fig:arrowpicmodel}. Throughout this section, the height will refer to $z=\rho+a$, and in Section~\ref{sec3} we will, with a slight abuse of terminology, refer to $a$ as the height.

Assign spectral parameters $\mu_i$ to each vertical line, and $\pm\lambda_i$ to each horizontal double line. The value is $-\lambda_i$ on the lower part of the double line and shifts to $\lambda_i$ on the upper part. In Fig.~\ref{fig:arrowpicmodel}, we write these parameters at the lines. Also define a fixed boundary parameter $\zeta$, associated to~the~reflecting wall at the turns.
To each vertex and each turn we assign a local weight
\begin{gather*}
\begin{aligned}
&a_+(\lambda, q^z)=a_-(\lambda, q^z)=\frac{[\lambda+1]}{[1]},\qquad
\\
&b_+(\lambda, q^z)=\frac{[\lambda][z-1]}{[z][1]}, &\quad&b_-(\lambda, q^z)=\frac{[\lambda][z+1]}{[z][1]},
\\
&c_+(\lambda, q^z)=\frac{[z+\lambda]}{[z]}, &&c_-(\lambda, q^z)=\frac{[z-\lambda]}{[z]},
\\
&k_+\big(\lambda, q^z, q^\zeta\big)=\frac{[z+\zeta-\lambda]}{[z+\zeta+\lambda]}, && k_-\big(\lambda, q^z, q^\zeta\big)=\frac{[\zeta-\lambda]}{[\zeta+\lambda]}.
\end{aligned}
\end{gather*}
These functions correspond to the local states as in Fig.~\ref{fig:6vmodel} and in Fig.~\ref{fig:reflectingends}.
The functions are well-defined: it is clear that $[z+1/\eta]=-[z]$ and $[\zeta+1/\eta]=-[\zeta]$, but when translating $z$ or~$\zeta$ by $1/\eta$, the numerator and denominator of the weights change simultaneously, so the minus signs cancel.

Sometimes, when we are only interested in the spin configurations around a vertex, or when~$\lambda$ and $z$ are clear, we will refer to a $w$ vertex, meaning a vertex with weight $w\big(\lambda, q^z\big)$, where $w$ is one of~$a_\pm$, $b_\pm$ or~$c_\pm$. Similarly a~$k_\pm$ turn will refer to a turn with weight $k_\pm\big(\lambda, q^z, q^\zeta\big)$, when~$\lambda$,~$z$ and~$\zeta$ are clear, or when it is the direction of the spin on the turning edge that is of importance. We~will also use the term positive (negative) turn for a~$k_+$ $(k_-)$ turn.

The local weight at a vertex with the positive directions up and to the right depends on the spins of the surrounding edges, but also on the height $z$ on the face to the upper left, as well as on the difference between the spectral parameters on the incoming lines from the left and the bottom. Because of the reflecting ends, we need to differentiate between the vertices on~the left oriented and the right oriented horizontal lines. The vertices in the right oriented rows are depicted in Fig.~\ref{fig:arrowpic}, and the vertices in the left oriented rows are the same, tilted 90 degrees counterclockwise, as in Fig.~\ref{fig:nodeupdown}. The (local) weight of the vertex in Fig.~\ref{fig:upperrow} is $w\big(\lambda_i-\mu_j, q^z\big)$, and for the vertex in Fig.~\ref{fig:lowerrow}, the weight is $w\big(\lambda_i+\mu_j, q^z\big)$, where $w$ is one of~$a_\pm$, $b_\pm$ or $c_\pm$.

\looseness=1 The boundary weight at each turn depends on the spin on the turning edge, but also on the spectral parameter $\lambda_i$ of the line going through the turn, and the height on the face outside the turn, as in Fig.~\ref{fig:reflectingends}. The weight also depends on the boundary parameter $\zeta$ which is fixed. The~height outside the turn is the same for all turns in the 8VSOS model with reflecting end.

Defining the height in the upper left corner to be $\rho$, the weight at a vertex is always \mbox{$w\big(\lambda_i\pm\mu_j, q^{\rho+a}\big)$}, for some $a\in \Z$, and the weight at a turn is always $k_\pm\big(\lambda_i, q^\rho, q^\zeta\big)$. The weight of a state is the product of all local weights of the vertices and the turns.

On the left side of the model we have the reflecting wall. It~remains to impose boundary conditions to the remaining three sides of the model. For these sides, we take the domain wall boundary conditions (DWBC), which in this case means that the ingoing edges at the bottom and the outgoing edges at the right have spin $1$, and the ingoing edges at the right and the outgoing edges at the top have spin~$-1$. This means that the lattice has arrows pointing inwards on the top and the bottom edges, and arrows pointing outwards on the edges to the right, as in Fig.~\ref{fig:arrowpicmodel}. If the height in the upper left corner is defined to be $\rho$, all the heights of the~faces at the boundaries are determined by the boundary conditions, except for the heights of the faces inside the loops.

\subsection{The partition function}
Let $w(\text{vertex})$ be one of the local weights $a_\pm\big(\lambda_i\pm \mu_j, q^{z}\big)$, $b_\pm\big(\lambda_i\pm\mu_j, q^{z}\big)$, $c_\pm\big(\lambda_i\pm\mu_j, q^{z}\big)$
at a~vertex with height $z=\rho+a_{\text{vertex}}$ in the upper left face, and let $w(\text{turn})$ be the local weight at~one of the turns, given by one of the weights $k_\pm\big(\lambda_i, q^\rho, q^\zeta\big)$.
The partition function of the 8VSOS model with DWBC and reflecting end is
\begin{gather*}
Z_n\big(q^{\lambda_1}, \dots, q^{\lambda_n}, q^{\mu_1}, \dots, q^{\mu_n}, q^\rho, q^\zeta\big)= \sum_{\text{states}} \prod_{\text{vertices}} w(\text{vertex}) \prod_{\text{turns}} w(\text{turn})
\\\hphantom{Z_n(q^{\lambda_1}, \dots, q^{\lambda_n}, q^{\mu_1}, \dots, q^{\mu_n}, q^\rho, q^\zeta)}
{}=\sum_{\text{states}} \prod_{\text{vertices}} w\big(\lambda_i\pm \mu_j, q^{\rho+a_\text{vertex}}\big)\prod_{\text{turns}} w\big(\lambda_i, q^\rho, q^\zeta\big).
\end{gather*}
The partition function also depends on $\tau$ and $\eta$.

\begin{figure}[t]
\centering
	\subfloat[$w(\lambda_i-\mu_j, q^z)$\label{fig:upperrow}]{%
	\begin{tikzpicture}[scale=0.75, font=\footnotesize]
		\draw[->] (0,1) node[left]{$\lambda_i$} -- (2,1);
		\draw[->] (1,0) node[below]{$\mu_j$} -- (1,2);
		\node at (0.5, 1.5) {$z$};
		\node at (-1,0) {\phantom{$\bullet$}};
		\node at (3,0) {\phantom{$\bullet$}};
	\end{tikzpicture}
	}\hfil
	\subfloat[$w(\lambda_i+\mu_j, q^z)$\label{fig:lowerrow}]{%
	\begin{tikzpicture}[scale=0.7, font=\footnotesize]
		\draw[<-] (0,1) -- (2,1) node[right]{$-\lambda_i$};
		\draw[->] (1,0) node[below]{$\mu_j$} -- (1,2);
		\node at (0.5, 0.5) {$z$};
		\node at (-1,0) {\phantom{$\bullet$}};
		\node at (3,0) {\phantom{$\bullet$}};
	\end{tikzpicture}
	}\\
	\caption{The different vertex weights depending on the direction of the row in the 8VSOS model with reflecting end, with spectral parameters $\lambda_i$ and $\mu_j$ and height $z$.}
	\label{fig:nodeupdown}
\end{figure}
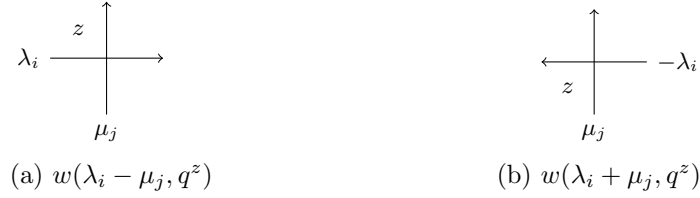

{\samepage
To see that the partition function is well-defined, we need to make sure that it is invariant under translations of~$\lambda_i$, $\mu_j$, $\rho$ and $\zeta$ with $1/\eta$. For $\rho$ and $\zeta$, this is clear, since the weights are well-defined. It~holds that $[x+1/\eta]=-[x]$. In the vertex weights, $\lambda_i$ and $\mu_j$ show up only in the numerators. Luckily, in each state, there are always two weights $w\big(\lambda_i\pm \mu_j, q^z\big)$ given by~each pair $\lambda_i$ and $\mu_j$. Thus, in the partition function, a translation of~$\lambda_i$ or $\mu_j$ will affect an even number of factors, so the minus signs will cancel each other. Hence the partition function is also invariant under translations of any $\lambda_i$ and $\mu_j$ by $1/\eta$.

}

In \cite{Filali2011}, Filali obtained a determinant formula for the partition function of the 8VSOS model with DWBC and reflecting end, namely,
\begin{gather}
Z_n\big(q^{\lambda_1}, \dots, q^{\lambda_n}, q^{\mu_1}, \dots, q^{\mu_n}, q^\rho, q^\zeta\big)\nonumber
\\ \qquad{}
=[1]^{n-2n^2} \prod_{i=1}^n \frac{[2\lambda_i][\zeta-\mu_i][\rho+\zeta+\mu_i][\rho+(2i-n-2)]}
{[\zeta+\lambda_i][\rho+\zeta+\lambda_i][\rho+(n-i)]}\nonumber
\\ \qquad\phantom{=}{}
\times \frac{\prod\limits_{i,j=1}^n [\lambda_i+\mu_j+1][\lambda_i-\mu_j+1][\lambda_i+\mu_j][\lambda_i-\mu_j]}{\prod\limits_{1\leq i<j\leq n} [\lambda_i+\lambda_j+1][\lambda_i-\lambda_j][\mu_j+\mu_i][\mu_j-\mu_i]} \det_{1\leq i,j\leq n} K_{ij},\label{Filalisdeterminantformula}
\end{gather}
where
\begin{gather*}
K_{ij}=\frac{1}{[\lambda_i+\mu_j+1][\lambda_i-\mu_j+1][\lambda_i+\mu_j][\lambda_i-\mu_j]}.
\end{gather*}

\subsection{The three-color model}
\label{3cmodel}

The three-color model is a model on a square lattice, with the faces filled with three different colors, which we call color $0$, $1$, and $2$, such that adjacent faces have different colors. A weight~$t_i$ is assigned to each face of color $i$. A state of the three-color model is called a three-coloring.

We study the three-color model on the $2n\times n$ lattice (i.e. a lattice with $(2n+1)\times(n+1)$ faces).
If we reduce the heights $\rho+a$ of the faces in the 8VSOS model to $a \operatorname{mod} 3$, the states of the 8VSOS model can be identified with the states of the three-color model (see Fig.~\ref{fig:3colormodel}).
The DWBC and the reflecting end in the 8VSOS model correspond to the following rules for the colors in the three-color model. In the upper left corner, we fix color $0$. On three of the boundaries, the colors alternate cyclically. Starting from the upper left corner, going to the right, the colors increase in the order $0<1<2<0$, to reach $n \operatorname{mod} 3$ in the upper right corner. From there, going down, the colors decrease down to $-n\operatorname{mod} 3$ in the lower right corner. Continuing to the left, the colors increase again, up to $0$ in the lower left corner.
On the left side, at the reflecting wall, every second face has color $0$. Inside the turns, the colors differ depending on the type of turn in the corresponding state of the 8VSOS model. A negative turn corresponds to color $1$, and a positive turn corresponds to color $2$. We~will henceforth assume these boundary conditions, even if we do not mention them explicitly.
The partition function of the three-color model, with DWBC and reflecting end, and with color $0$ fixed in the upper left corner is
\begin{gather*}
Z_n^{3C}(t_0, t_1, t_2)=\sum_{\text{states}} \prod_{\text{faces}} t_i.
\end{gather*}

Let $m$ be the number of positive turns in a state of the 8VSOS model with DWBC and reflecting end. Specifying $m$ means that we have specified the number of faces with color $2$ on the left side. If we specify $m=0$, the colors on the left side alternate between color $0$ and $1$. There is a bijection between the three-colorings with $m=0$ and the VSASMs of size $(2n+1)\times (2n+1)$ \cite{Kuperberg2002}.

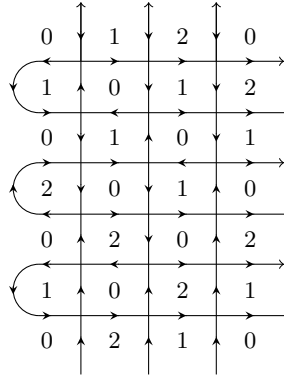
\begin{figure}[t]
\vspace{3mm}
\centering
\begin{tikzpicture}[scale=0.9, font=\footnotesize]
	\foreach \y in {1,...,3} {
		\draw[midarrow={stealth}] (3.55,1.5*\y-.25-.38) -- +(0.01,0);
		\draw[midarrow={stealth}] (3.55,1.5*\y-.25+.38) -- +(0.01,0);
		
		\draw (1, 1.5*\y-.25-.38) -- +(3, 0);
		\draw[->] (1, 1.5*\y-.25+.38) -- +(3, 0);
		
		\draw (0.38,1.5*\y-.25+.38) arc (90:270:0.38);
		}
		
	\foreach \x in {1,...,3} {
		\draw (\x,0) -- +(0,.87); 
			\draw[midarrow={stealth}] (\x,0.55) -- +(0,0.01);
		\draw[->] (\x,4.63) -- +(0,.87);	
			\draw[midarrow={stealth reversed}] (\x,4.63+.43) -- +(0,0.01);	
			
		\draw (\x,.87) -- +(0,4.63);
	}
	
		\draw[-stealth reversed] (0,1.55-0.25) -- (0,1.56-0.25);
		\draw(.38,1.25-.38) -- (1,1.25-.38);
			\draw[midarrow={stealth}] (.55,1.25-.38) -- +(0.01,0);
		\draw(.38,1.25+.38) -- (1,1.25+.38);
			\draw[midarrow={stealth reversed}] (.55,1.25+.38) -- +(0.01,0);
		\node at (0.5, 1.25) {$1$};
		
		\draw[-stealth] (0,3.05-0.25) -- (0,3.06-0.25);
		\draw (.38,2.75-.38) -- (1,2.75-.38);
			\draw[midarrow={stealth reversed}] (.55,2.75-.38) -- +(0.01,0);
		\draw(.38,2.75+.38) -- (1,2.75+.38);
			\draw[midarrow={stealth}] (.55,2.75+.38) -- +(0.01,0);
		\node at (0.5, 2.75) {$2$};
		
		\draw[-stealth reversed] (0,4.45-0.25) -- (0,4.56-0.25);
		\draw(.38,4.25-.38) -- (1,4.25-.38);
			\draw[midarrow={stealth}] (.55,4.25-.38) -- +(0.01,0);
		\draw (.38,4.25+.38) -- (1,4.25+.38);
			\draw[midarrow={stealth reversed}] (.55,4.25+.38) -- +(0.01,0);
		\node at (0.5, 4.25) {$1$};
		
		\draw [midarrow={stealth}] (1.55, 1.5-.25-.38) -- +(0.01, 0);
		\draw [midarrow={stealth}] (2.55, 1.5-.25-.38) -- +(0.01, 0);
		\draw [midarrow={stealth reversed}](1.55, 1.5-.25+.38) -- +(0.01, 0);
		\draw [midarrow={stealth}] (2.55, 1.5-.25+.38) -- +(0.01, 0);
		
		\draw [midarrow={stealth}] (1.55, 3-.25-.38) -- +(0.01, 0);
		\draw [midarrow={stealth}] (2.55, 3-.25-.38) -- +(0.01, 0);
		\draw [midarrow={stealth}] (1.55, 3-.25+.38) -- +(0.01, 0);
		\draw [midarrow={stealth reversed}](2.55, 3-.25+.38) -- +(0.01, 0);
		
		\draw [midarrow={stealth reversed}] (1.55, 4.5-.25-.38) -- +(0.01, 0);
		\draw [midarrow={stealth}] (2.55, 4.5-.25-.38) -- +(0.01, 0);
		\draw [midarrow={stealth}](1.55, 4.5-.25+.38) -- +(0.01, 0);
		\draw [midarrow={stealth}] (2.55, 4.5-.25+.38) -- +(0.01, 0);
		
			\draw [midarrow={stealth}] (1,0.87+.43) -- +(0,0.01);
			\draw [midarrow={stealth}] (1,1.63+.43) -- +(0,0.01);
			\draw [midarrow={stealth reversed}] (1,2.37+.43) -- +(0,0.01);
			\draw [midarrow={stealth reversed}] (1,3.13+.43) -- +(0,0.01);
			\draw [midarrow={stealth}] (1,3.87+.43) -- +(0,0.01);
			
			\draw [midarrow={stealth}] (2,0.87+.43) -- +(0,0.01);
			\draw [midarrow={stealth reversed}] (2,1.63+.43) -- +(0,0.01);
			\draw [midarrow={stealth reversed}] (2,2.37+.43) -- +(0,0.01);
			\draw [midarrow={stealth}] (2,3.13+.43) -- +(0,0.01);
			\draw [midarrow={stealth reversed}] (2,3.87+.43) -- +(0,0.01);
			
			\draw [midarrow={stealth}] (3,0.87+.43) -- +(0,0.01);
			\draw [midarrow={stealth}] (3,1.63+.43) -- +(0,0.01);
			\draw [midarrow={stealth}] (3,2.37+.43) -- +(0,0.01);
			\draw [midarrow={stealth reversed}] (3,3.13+.43) -- +(0,0.01);
			\draw [midarrow={stealth reversed}] (3,3.87+.43) -- +(0,0.01);
			
		\node at (0.5, 5) {$0$};
 \node at (1.5, 5) {$1$};
 \node at (2.5, 5) {$2$};
 \node at (3.5, 5) {$0$};

 \node at (3.5, 4.25) {$2$};
 \node at (3.5, 3.5) {$1$};
 \node at (3.5, 2.75) {$0$};
 \node at (3.5, 2) {$2$};
 \node at (3.5, 1.25) {$1$};

 \node at (3.5, 0.5) {$0$};
 \node at (2.5, 0.5) {$1$};
 \node at (1.5, 0.5) {$2$};
 \node at (0.5, 0.5) {$0$};
		
		\node at (0.5, 2) {$0$};
		\node at (0.5, 3.5) {$0$};
		
		\node at (1.5, 4.25) {$0$};
 \node at (1.5, 3.5) {$1$};
 \node at (1.5, 2.75) {$0$};
 \node at (1.5, 2) {$2$};
 \node at (1.5, 1.25) {$0$};
		
		\node at (2.5, 4.25) {$1$};
 \node at (2.5, 3.5) {$0$};
 \node at (2.5, 2.75) {$1$};
 \node at (2.5, 2) {$0$};
 \node at (2.5, 1.25) {$2$};
\end{tikzpicture}
\vspace{-1mm}
\caption{A state of the three-color model, for $n=3$, with colors $0$, $1$ and $2$. The arrows on the edges show the corresponding state in the 8VSOS model.}
\label{fig:3colormodel}
\end{figure}

\section{Rewriting of the partition function}\label{sec3}
In his proof of the alternating sign matrix conjecture, Kuperberg studied the partition function of the 6V model with DWBC with $\lambda_i=-1/2$ and $\mu_j=0$, for all $i$, $j$, and $\eta=-2/3$, so that $q$ becomes a cubic root of unity. In this section, we specialize to these values. Following the proof of Lemma~7.1 in~\cite{Rosengren2009}, we simplify the expression for the partition function of the 8VSOS model with reflecting end. We~find a way to express the partition function in terms of the heights of the faces, rather than in terms of the vertex weights. In this way, it corresponds to the partition function of the three-color model. Finally we write the partition function as a sum over the number of positive turns, to be able to compare factors term by term in Section~\ref{sec4}.
	
We will need the following result on the number of different vertex types.
\begin{Lemma}
\label{numberofnodesoftypew}
For any given state of the 8VSOS model with DWBC and reflecting end, let $\nu(w)$ be the number of vertices or turns of type $w$. Then we have
\begin{gather*}
\nu(b_+)=\nu(b_-)+\binom{n+1}{2} \qquad \text{and} \qquad \nu(c_+)+2\nu(k_-)=\nu(c_-)+n.
\end{gather*}
\end{Lemma}

As in the proof of the corresponding result for the 6V model without reflecting end (see,~e.g., Section~7.1 of \cite{Bressoud1999}), we will count arrows. In our case we need to differentiate between the vertices of the ingoing and outgoing rows.
For the proof, define $a_+^N$ and $a_+^S$ ($N$ for north and $S$ for south) for the upper and lower parts of the double rows respectively (as in Fig.~\ref{fig:aplussandn}), such that $\nu(a_+)=\nu\big(a_+^N\big)+\nu\big(a_+^S\big)$. Define the weights similarly for all other vertex types. For instance, we see that $a_+^S$ has two left and two up pointing arrows, whereas $a_+^N$ has two right and two up pointing arrows. We~interpret the $k_+$ turn as one left arrow, one up arrow and one right arrow, and reversed for the $k_-$ turn (see Fig.~\ref{fig:arrowpicreflend}).

\begin{figure}[t]
\centering
 \subfloat{%
 \begin{tikzpicture}[baseline={([yshift=-.5*10pt*0.6]current bounding box.center)}, scale=0.7, font=\footnotesize]
		\draw[midarrow={stealth}] (0,1) -- (1,1); 
		\draw[midarrow={stealth reversed}, <-] (2,1) -- (1,1); 
		\draw[midarrow={stealth}] (1,0) -- (1,1); 
		\draw[midarrow={stealth reversed}, <-] (1,2) -- (1,1); 
		\draw (1,-1) node{\small{$a_+^N$}};
		\end{tikzpicture}
	}\hfil
	 \subfloat{%
 \begin{tikzpicture}[baseline={([yshift=-.5*10pt*0.6]current bounding box.center)}, scale=0.7, font=\footnotesize]
		\draw[midarrow={stealth reversed}, <-] (0,1) -- (1,1); 
		\draw[midarrow={stealth}] (2,1) -- (1,1); 
		\draw[midarrow={stealth}] (1,0) -- (1,1); 
		\draw[midarrow={stealth reversed}, <-] (1,2) -- (1,1); 
			\draw (1,-1) node{\small{$a_+^S$}};
		\end{tikzpicture}
		}
\vspace{-1mm}
\caption{The vertex $a_+$ in the upper and lower part of a double row respectively.}
\label{fig:aplussandn}
\end{figure}
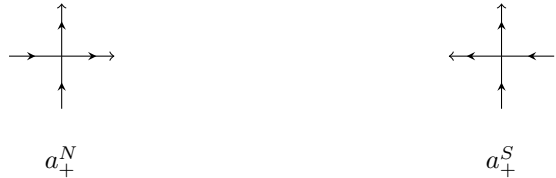

\begin{proof}[Proof of Lemma~\ref{numberofnodesoftypew}.]
Since the number of arrows pointing down on the upper boundary is the same as the number of arrows pointing up on the lower boundary, and all arrows have to ``travel through the lattice'' according to the ice rule,
 and go out to the right, the total number of arrows pointing upwards must be the same as the number of arrows pointing downwards. Therefore the number of arrows pointing upwards in a state is always $n(n+1)$. The same holds for the number of arrows pointing downwards.
A similar reasoning yields that the total number of arrows pointing to the left is $\sum_{i=1}^n i=n(n+1)/2$ and the total number of arrows pointing to~the right is $\sum_{i=n}^{2n} i=3n(n+1)/2$.

On the other hand, we get the number of up arrows by, on every second row, counting the number of vertex types with up arrows (vertex types with two up arrows counted twice) plus the~number of~$k_+$ turns and possibly compensate for all the arrows on the lower boundary, depending on which rows we counted.
Hence we get that the number of up arrows for any state~is
\begin{gather}
\label{equ2}
2\nu\big(a_+^S\big) + 2\nu\big(b_+^S\big) + \big(c_+^S\big) + \big(c_-^S\big) + \nu(k_+)=n(n+1)
\end{gather}
if we count arrows at every ingoing (lower) row, and
\begin{gather}
\label{equ1}
2\nu\big(a_+^N\big) + 2\big(b_-^N\big) + \nu\big(c_+^N\big) + \nu\big(c_-^N\big) + \nu(k_+) + n=n(n+1)
\end{gather}
if we instead count arrows at the outgoing (upper) rows.
Similarly for the down arrows, we get
\begin{gather}
\label{equ4}
2\nu\big(a_-^S\big) + 2\big(b_-^S\big) + \big(c_+^S\big) + \big(c_-^S\big) + \nu(k_-)+n=n(n+1)
\end{gather}
and
\begin{gather}
\label{equ3}
2\nu\big(a_-^N\big) + 2\big(b_+^N\big) + \nu\big(c_+^N\big) + \nu\big(c_-^N\big) + \nu(k_-)=n(n+1).
\end{gather}
To get the number of left arrows, we count the number of the different vertices with left arrows, plus the number of all turns. In this way, we count every left arrow twice. Hence the number of left arrows is
\begin{gather}
\label{equ5}
\frac{2\big[\nu\big(a_-^N\big) + \big(b_-^N\big) + \nu\big(a_+^S\big) + \big(b_-^S\big)\big] + \nu(c_+) + \nu(c_-) +n}{2}
=\frac{n(n+1)}{2}.
\end{gather}
Similarly we get the number of right arrows by counting the number of different vertices of right arrows, plus the number of turns, plus the number of arrows on the right boundary. Now we have counted all arrows twice. The number of right arrows is
\begin{gather}
\label{equ6}
\frac{2\big[\nu\big(a_+^N\big) + \big(b_+^N\big) + \nu\big(a_-^S\big) + \nu\big(b_+^S\big)\big] + \nu(c_+) + \nu(c_-)+3n}{2}
=\frac{3n(n+1)}{2}.
\end{gather}
Now we add the equations \eqref{equ2}, \eqref{equ3} and two times \eqref{equ6}, and subtract \eqref{equ4}, \eqref{equ1} and two times \eqref{equ5}, to get
\begin{gather*}
2[\nu(b_+)-\nu(b_-)]=n(n+1),
\end{gather*}
which is the first part of the lemma.

To obtain the second result, we consider the left and right arrows in each row. For each pair of rows connected by a $k_+$ turn, the lower row has a left arrow $<$ on the leftmost edge and a~right arrow $>$ on the rightmost edge, as in the left picture of Fig.~\ref{fig:onedoublerow}. In between we can have any combination of arrows $<\cdots >$. Every two consecutive arrows $<<$ or $>>$ correspond to $a_\pm^S$ or $b_\pm^S$ vertices. Every $<>$ corresponds to a $c_+^S$ vertex and every $><$ is a $c_-^S$ vertex, so considering only the $c$ vertices, the first and last vertices are $c_+^S$ vertices. The upper row starts and ends with right arrows, so we have $>\cdots >$. Here $<>$ is a $c_-^N$ vertex and $><$ is a $c_+^N$ vertex. This means that the upper row starts with a $c_+^N$ vertex and ends with a $c_-^N$ vertex. Hence
$\big(c_+^S\big)=\big(c_-^S\big)+1$ and $\nu\big(c_+^N\big)=\nu\big(c_-^N\big)$ at the rows with a $k_+$ turn. Similarly $\big(c_+^S\big)=\big(c_-^S\big)$ and $\nu\big(c_+^N\big)=\nu\big(c_-^N\big)-1$ at the rows with a $k_-$ turn.
Hence for the whole lattice,
\begin{gather*}
\nu(c_+)=\nu(c_-)+n-2\nu(k_-),
\end{gather*}
which is the second part of the lemma.
\end{proof}

\begin{figure}[t]
\centering
\subfloat{%
\begin{tikzpicture}[scale=0.7]

		\draw (0.4,0.4) arc (90:270:0.4);
		\draw[midarrow={stealth reversed}] (.4,-.4) -- (1,-.4);
		\draw (1, -.4) -- +(2, 0);
		\draw[midarrow={stealth}] (.4,+.4) -- (1,+.4);
		\draw[-] (1, +.4) -- +(2, 0);
		\draw[-stealth] (0,0.05) -- (0,0.06);
		\draw[midarrow={stealth}] (3,-.4) -- +(1,0);
		\draw[midarrow={stealth},->] (3,+.4) -- +(1,0);

	\foreach \x in {1,...,3} {
 \draw (\x,-1.25) -- +(0,0.85);
		\draw (\x,-0.4) -- +(0,0.8); 
		\draw[->] (\x,0.4) -- +(0,.85);			
	}
\end{tikzpicture}
	}\hfil
	\subfloat{%
	\begin{tikzpicture}[scale=0.7]

		\draw (0.4,0.4) arc (90:270:0.4);
		\draw[midarrow={stealth }] (.4,-.4) -- (1,-.4);
		\draw (1, -.4) -- +(2, 0);
		\draw[midarrow={stealth reversed}] (.4,+.4) -- (1,+.4);
		\draw[-] (1, +.4) -- +(2, 0);
		\draw[-stealth] (0,-0.05) -- (0,-0.06);
		\draw[midarrow={stealth}] (3,-.4) -- +(1,0);
		\draw[midarrow={stealth},->] (3,+.4) -- +(1,0);

	\foreach \x in {1,...,3} {
 \draw (\x,-1.25) -- +(0,0.85);
		\draw (\x,-0.4) -- +(0,0.8); 
		\draw[->] (\x,0.4) -- +(0,.85);			
	}
\end{tikzpicture}
}\\
\vspace{-1mm}
\caption{On the left, a double row with a $k_+$ turn, and on the right, a double row with a $k_-$ turn, for $n=3$.}
\label{fig:onedoublerow}
\end{figure}
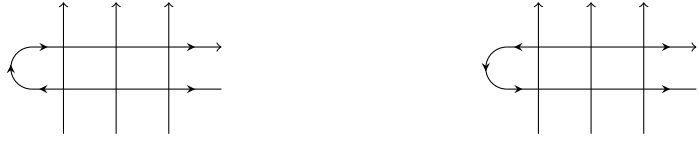

\subsection{The partition function in terms of the local heights}
\label{sectionchangepartitionfunction}

To simplify writing, we start by doing the variable change $q^{\rho} \rightarrow \rho$ and $q^\zeta \rightarrow \zeta$. Then we specialize $\lambda_i=-1/2$ and $\mu_i=0$, following Kuperberg.
	
\begin{Proposition}\label{prop3.2}
Let $\lambda_i=-1/2$, and $\mu_i=0$. For each state, let $N$ be the number of~$c_-$ vertices and $M$ the number of~$k_-$ turns. For each vertex, let $a$, $b$, $c$, $d$ denote the heights on the adjacent faces as in Fig.~{\rm \ref{fig:faceheights}}, and for each turn, let $a$ be the height inside the turn. Then the partition function of the 8VSOS model with DWBC and reflecting end is
\begin{gather*}
Z_n\big(q^{-1/2}, \dots, q^{-1/2}, 1,\dots, 1, \rho, \zeta\big)=C\sum_{\textup{states}} D\prod_{\textup{vertices}} \frac{\vartheta\big(\rho q^{(3a-b+3c-d)/4}\big)}{\vartheta(\rho q^a)}
\\ \hphantom{Z_n\big(q^{-1/2}, \dots, q^{-1/2}, 1,\dots, 1, \rho, \zeta\big)=}
{}\times\prod_{\textup{turns}}\left(\!\left( \frac{\vartheta\big(q^{1/2}\big)}{\vartheta(q)}\right)^{a+1} \frac{\vartheta\big(\rho^{(1-a)/2} \zeta q^{1/2}\big)}{\vartheta\big(\rho^{(1-a)/2} \zeta q^{-1/2}\big)}\right)\!,
\end{gather*}
where
\begin{gather*}
C=(-1)^{\binom{n+1}{2}}q^{(3n^2-n)/4}\left(\frac{\vartheta(q^{1/2})}{\vartheta(q)}\right)^{2n^2-n},
\qquad 
D=\left(\frac{q^{-1/2}\vartheta(q)^2}{\vartheta\big(q^{1/2}\big)^2}\right)^N
\left(\frac{\vartheta\big(q^{1/2}\big)}{\vartheta(q)}\right)^{2M}.
\end{gather*}
\end{Proposition}

The proof is similar to the proof of Theorem 7.1 in~\cite{Rosengren2009}.

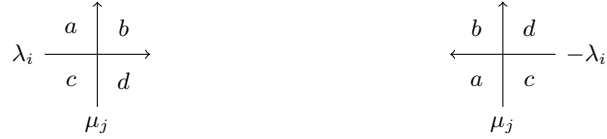
\begin{figure}[t]
\centering
	\subfloat{%
	\begin{tikzpicture}[scale=0.7, font=\footnotesize]
		\draw[->] (0,1) node[left]{$\lambda_i$} -- (2,1);
		\draw[->] (1,0) node[below]{$\mu_j$} -- (1,2);
		\node at (0.5, 1.5) {$a$};
		\node at (1.5, 1.5) {$b$};
		\node at (0.5, 0.5) {$c$};
		\node at (1.5, 0.5) {$d$};
	\end{tikzpicture}
	}\hfil
	\subfloat{%
\begin{tikzpicture}[scale=0.7, font=\footnotesize]
		\draw[<-] (0,1) -- (2,1) node[right]{$-\lambda_i$};
		\draw[->] (1,0) node[below]{$\mu_j$} -- (1,2);
		\node at (0.5, 1.5) {$b$};
		\node at (1.5, 1.5) {$d$};
		\node at (0.5, 0.5) {$a$};
		\node at (1.5, 0.5) {$c$};
	\end{tikzpicture}
}
\vspace{-2mm}
\caption{Vertices with heights $a$, $b$, $c$ and $d$ on the adjacent faces.}
\label{fig:faceheights}
\end{figure}

\begin{proof}
Each vertex is one of the vertices in Fig.~\ref{fig:faceheights}. Hence each weight is always $w(\lambda_i\pm \mu_j, \rho q^a)$.
Putting $\lambda_i=-1/2$ and $\mu_i=0$ yields that the weights at the vertices are always $w(-1/2, \rho q^a)$, and the partition function will be
\begin{gather*}
Z_n\big(q^{-1/2}, \dots, q^{-1/2}, 1,\dots, 1, \rho, \zeta\big)=\sum_{\text{states}} \prod_{\text{vertices}} w(-1/2, \rho q^a)\prod_{\text{turns}} k_\pm(-1/2, \rho,\zeta).
\end{gather*}
The local weights become
\begin{alignat*}{3}
&a_+(-1/2, \rho q^a)=a_-(-1/2, \rho q^a)\!=\!\frac{q^{1/4} \vartheta\big(q^{1/2}\big)}{\vartheta(q)},&&&
\\
&b_+(-1/2, \rho q^a)=
\frac{-q^{3/4}\vartheta\big(q^{1/2}\big)\vartheta(\rho q^{a-1})}{\vartheta(\rho q^a)\vartheta(q)},
&&b_-(-1/2, \rho q^a)=
\frac{-q^{-1/4}\vartheta\big(q^{1/2}\big)\vartheta\big(\rho q^{a+1}\big)}{\vartheta(\rho q^a)\vartheta(q)},&
\\
&c_+(-1/2, \rho q^a)=
\frac{q^{1/4}\vartheta\big(\rho q^{a-1/2}\big)}{\vartheta(\rho q^a)},
&&c_-(-1/2, \rho q^a)=
\frac{q^{-1/4}\vartheta\big(\rho q^{a+1/2}\big)}{\vartheta(\rho q^a)},&
\\
&k_+(-1/2, \rho,\zeta)=
\frac{q^{-1/2}\vartheta\big(\rho \zeta q^{1/2}\big)}{\vartheta\big(\rho \zeta q^{-1/2}\big)},
&&k_-(-1/2, \rho, \zeta)=
\frac{q^{-1/2}\vartheta\big(\zeta q^{1/2}\big)}{\vartheta\big(\zeta q^{-1/2}\big)},&
\end{alignat*}
where we used $q^{1/4}\vartheta\big(q^{-1/2}\big)=-q^{-1/4}\vartheta\big(q^{1/2}\big)$ to get $b_\pm(-1/2, \rho q^a)$.

Each term of the partition function consists of~$2n^2$ factors of weights of the vertices and $n$ factors of weights of the turns. From each vertex weight we take out a factor $q^{1/4} \vartheta\big(q^{1/2}\big)/\vartheta(q)$, and from each $k_\pm(-1/2, \rho, \zeta)
$ we take out the factor $q^{-1/2}$ and put in a prefactor.
Then we factor out $-q^{1/2}$ from each $b_+$ vertex, and $-q^{-1/2}$ from each $b_-$ vertex. By Lemma~\ref{numberofnodesoftypew}, there are always $\binom{n+1}{2}$ more $b_+$ vertices than $b_-$ vertices in each state. Hence some of these factors cancel each other, and $\big({-}q^{1/2}\big)^{n(n+1)/2}$ goes to the prefactor.
Let $N$ be the number of~$c_-$ vertices and $M$ the number of~$k_-$ turns in a given state. Lemma~\ref{numberofnodesoftypew} yields that the number of~$c_+$ vertices is $N+n-2M$. We~factor out $\frac{\vartheta(q)}{\vartheta(q^{1/2})}$ from each $c_+$ and $\frac{q^{-1/2}\vartheta(q)}{\vartheta(q^{1/2})}$ from each $c_-$,
so~that $\left(\frac{\vartheta(q)}{\vartheta(q^{1/2})}\right)^{n-2M}\left(\frac{q^{-1/2}\vartheta(q)^2}{\vartheta(q^{1/2})^2}\right)^N$ becomes a part of the prefactor.

Our new weights are
\begin{gather*}
\begin{aligned}
&\tilde a_+(-1/2, \rho q^a)=\tilde a_-(-1/2, \rho q^a)=1,
\\
&\tilde b_+(-1/2, \rho q^a)=\frac{\vartheta\big(\rho q^{a-1}\big)}{\vartheta(\rho q^a)},
&\qquad& \tilde b_-(-1/2, \rho q^a)=\frac{\vartheta\big(\rho q^{a+1}\big)}{\vartheta(\rho q^a)},
\\
&\tilde c_+(-1/2, \rho q^a)=
\frac{\vartheta\big(\rho q^{a-1/2}\big)}{\vartheta(\rho q^a)},
&&\tilde c_-(-1/2, \rho q^a)=\frac{\vartheta\big(\rho q^{a+1/2}\big)}{\vartheta(\rho q^a)},
\\
&\tilde k_+(-1/2, \rho,\zeta)=
\frac{\vartheta\big(\rho \zeta q^{1/2}\big)}{\vartheta\big(\rho \zeta q^{-1/2}\big)},
&& \tilde k_-(-1/2, \rho, \zeta)=\frac{\vartheta\big(\zeta q^{1/2}\big)}{\vartheta\big(\zeta q^{-1/2}\big)}.
\end{aligned}
\end{gather*}
One can check that for each type of vertex $\tilde a_\pm$, $\tilde b_\pm$ and $\tilde c_\pm$ with heights $a$, $b$, $c$, $d$ on the adjacent faces, as in Fig.~\ref{fig:faceheights},
we have
\begin{gather*}
\tilde w(-1/2, \rho q^a)=\frac{\vartheta\big(\rho q^{(3a-b+3c-d)/4}\big)}{\vartheta(\rho q^a)}.
\end{gather*}
Furthermore
\begin{gather*}
\tilde k_\pm(-1/2,\rho,\zeta)=\frac{\vartheta\big(\rho^{(1-a)/2} \zeta q^{1/2}\big)}{\vartheta\big(\rho^{(1-a)/2} \zeta q^{-1/2}\big)},
\end{gather*}
where $a$ is the height of the face inside the turn.
The proposition follows.
\end{proof}

\subsection{The partition function in terms of three-colorings}\label{partfcnfaces}
Put $\eta = -2/3$ and define $\omega={\rm e}^{2\pi{\rm i} /3}$. Then $\omega=q=q^{-1/2}$. For any $a\in \Z$ and arbitrary $x$, we~have
\begin{gather}
\label{mod3}
\omega^a=\omega^{a+3}, \qquad 1+\omega+\omega^2=0,
\\
\label{omegaomega2}
\vartheta\big(p^{1/2}\omega\big)=\vartheta\big(p^{1/2}\omega^2\big),
\end{gather}
and
\begin{gather}\vartheta(x\omega^a)\vartheta\big(x\omega^{a+1}\big)\vartheta\big(x\omega^{a+2}\big)=\vartheta\big(x^3, p^3\big).
\label{threeproduct}
\end{gather}
Other identities we will use are \cite{Rosengren2011}
\begin{gather}\label{Rosengren4.2a}
\vartheta(-1)\vartheta\big(p^{1/2}\big)\vartheta\big({-}p^{1/2}\big)=2,
\end{gather}
and
\begin{gather}
\label{Rosengren4.2b}
\vartheta(-\omega)\vartheta\big(p^{1/2}\omega\big)\vartheta\big({-}p^{1/2}\omega\big)=-\omega^2.
\end{gather}

Now the constants from Proposition~\ref{prop3.2} become
$C=(-1)^{\binom{n}{2}}\omega^{n^2}$ and $D=\omega^M$.
Observe that if $x$ is an even number, then $q^{x/4}=q^x$. Inserting the possible values of the heights on the faces, we see that $3a-b+3c-d$ is always even. Since $b$ and $d$ are noncongruent modulo $3$, then $b$, $d$ and $-b-d$ are noncongruent modulo~$3$. Thus,
using \eqref{threeproduct},
\begin{gather*}
\frac{\vartheta\big(\rho q^{(3a-b+3c-d)/4}\big)}{\vartheta(\rho q^a)}
=\frac{\vartheta\big(\rho \omega^{-b-d}\big)}{\vartheta(\rho \omega^a)}
=\frac{\vartheta\big(\rho^3, p^3\big)}{\vartheta(\rho \omega^a)\vartheta(\rho \omega^b)\vartheta(\rho \omega^d)}.
\end{gather*}
Hence the partition function becomes
\begin{gather}
Z_n(\omega, \dots, \omega, 1,\dots, 1, \rho,\zeta)\nonumber
\\ \qquad{}
=C^\prime \sum_{\text{states}} \omega^M\prod_{\text{vertices}}\frac{1}{\vartheta(\rho \omega^a)\vartheta(\rho \omega^b)\vartheta(\rho \omega^d)}
\prod_{\text{turns}} \frac{\vartheta\big(\rho^{(1-a)/2} \zeta \omega^{-1}\big)}{\vartheta\big(\rho^{(1-a)/2} \zeta \omega\big)} ,
\label{partfcnabd}
\end{gather}
where $C^\prime=(-1)^{\binom{n}{2}}\omega^{n^2}\vartheta\big(\rho^3, p^3\big)^{2n^2}$.
This means that for each vertex in the lattice, we just need to know the heights of three of the adjacent faces, that is, the faces $a$, $b$ and $d$ in Fig.~\ref{fig:faceheights}.
\begin{figure}[t]
\centering
\subfloat{
\begin{tikzpicture}[scale=0.9, baseline={([yshift=-.5*10pt*0.6]current bounding box.center)}, font=\footnotesize]
	\foreach \y in {1,...,2} {
		\draw (.38,1.5*\y-.25-.38) -- +(1.3+0.32,0);
		\draw[midarrow={stealth}] (2,1.5*\y-.25-.38) -- +(1,0);
		\draw (.38,1.5*\y-.25+.38) -- +(1.3+0.32,0);
		\draw[midarrow={stealth},->] (2,1.5*\y-.25+.38) -- +(1,0);
		\draw (0.38,1.5*\y-.25+.38) arc (90:270:0.38);			
	}

	\fill[preaction={fill,white},pattern=north east lines, pattern color=gray] (0,0) rectangle (-.15,4) ; \draw (0,0) -- (0,4);
	
	\foreach \x in {1,...,2} {%
		\draw[midarrow={stealth}] (\x,0) -- +(0,.87); 
		\draw (\x,.87) -- +(0,2.26); 
		
		\draw[midarrow={stealth reversed}, ->] (\x,3.13) -- +(0,.97);	
	}

	\node at (0.5, 3.9) {$0$};
 \node at (1.5, 3.9) {$1$};
 \node at (2.5, 3.9) {$2$};

 \node at (2.85, 2.75) {$1$};
 \node at (2.85, 2) {$0$};
 \node at (2.7, 1.25) {$-1$};

 \node at (2.5, 0.1) {$-2$};
 \node at (1.5, 0.1) {$-1$};
 \node at (0.5, 0.1) {$0$};

		\node at (0.3, 2) {$0$};
		
	\node at (0.7, 3.4) {\huge $\cdot$};
	\node at (1.3, 3.4) {\huge $\cdot$};
	\node at (1.7, 3.4) {\huge $\cdot$};
	\node at (2.3, 3.4) {\huge $\cdot$};
	
	\node at (2.3, 2.9) {\huge $\cdot$};
	\node at (2.3, 2.6) {\huge $\cdot$};
	\node at (2.3, 2) {\huge $\cdot$};
	\node at (2.3, 1.4) {\huge $\cdot$};
	\node at (2.3, 1.1) {\huge $\cdot$};
	
	\node at (1.5, 0.6) {\huge $\cdot$};
 	\node at (0.7, 0.6) {\huge $\cdot$};
 	
	\node at (0.7, 1.25) {\huge $\cdot$};
	\node at (0.7, 1.85) {\huge $\cdot$};
	\node at (0.7, 2.15) {\huge $\cdot$};
	\node at (0.7, 2.75) {\huge $\cdot$};
	
	\node at (1.3, 1.10) {\huge $\cdot$};
	\node at (1.5, 1.40) {\huge $\cdot$};
	\node at (1.7, 1.10) {\huge $\cdot$};

	\node at (1.3, 1.85) {\huge $\cdot$};
	\node at (1.5, 2.15) {\huge $\cdot$};
	\node at (1.7, 1.85) {\huge $\cdot$};
	
	\node at (1.3, 2.60) {\huge $\cdot$};
	\node at (1.5, 2.90) {\huge $\cdot$};
	\node at (1.7, 2.60) {\huge $\cdot$};

\end{tikzpicture}
}\hfil
\subfloat{%
\begin{tikzpicture}[scale=0.9, baseline={([yshift=-.5*10pt*0.6]current bounding box.center)}, font=\footnotesize]
	\foreach \y in {1,...,3} {
		\draw (.38,1.5*\y-.25-.38) -- +(2.3+0.32,0);
		\draw[midarrow={stealth}] (3,1.5*\y-.25-.38) -- +(1,0);
		\draw (.38,1.5*\y-.25+.38) -- +(2.3+0.32,0);
		\draw[midarrow={stealth},->] (3,1.5*\y-.25+.38) -- +(1,0);
		\draw (0.38,1.5*\y-.25+.38) arc (90:270:0.38);
				
	}

	\fill[preaction={fill,white},pattern=north east lines, pattern color=gray] (0,0) rectangle (-.15,5.5) ; \draw (0,0) -- (0,5.5);
	
	\foreach \x in {1,...,3} {
		\draw[midarrow={stealth}] (\x,0) -- +(0,.87); 
		\draw (\x,.87) -- +(0,3.76); 
		
		\draw[midarrow={stealth reversed}, ->] (\x,4.63) -- +(0,.97);	
	}

	\node at (0.5, 5.4) {$0$};
 \node at (1.5, 5.4) {$1$};
 \node at (2.5, 5.4) {$2$};
 \node at (3.5, 5.4) {$3$};

 \node at (3.85, 4.25) {$2$};
 \node at (3.85, 3.5) {$1$};
 \node at (3.85, 2.75) {$0$};
 \node at (3.7, 2) {$-1$};
 \node at (3.7, 1.25) {$-2$};

 \node at (3.5, 0.1) {$-3$};
 \node at (2.5, 0.1) {$-2$};
 \node at (1.5, 0.1) {$-1$};
 \node at (0.5, 0.1) {$0$};
		
		\node at (0.3, 2) {$0$};
		\node at (0.3, 3.5) {$0$};

	\node at (0.7, 4.9) {\huge $\cdot$};
	\node at (1.3, 4.9) {\huge $\cdot$};
	\node at (1.7, 4.9) {\huge $\cdot$};
	\node at (2.3, 4.9) {\huge $\cdot$};
	\node at (2.7, 4.9) {\huge $\cdot$};
	\node at (3.3, 4.9) {\huge $\cdot$};
	
	\node at (3.3, 4.4) {\huge $\cdot$};
	\node at (3.3, 4.1) {\huge $\cdot$};
	\node at (3.3, 3.5) {\huge $\cdot$};
	\node at (3.3, 2.9) {\huge $\cdot$};
	\node at (3.3, 2.6) {\huge $\cdot$};
	\node at (3.3, 2) {\huge $\cdot$};
	\node at (3.3, 1.4) {\huge $\cdot$};
	\node at (3.3, 1.1) {\huge $\cdot$};
	
	\node at (2.5, 0.6) {\huge $\cdot$};
	\node at (1.5, 0.6) {\huge $\cdot$};
 	\node at (0.7, 0.6) {\huge $\cdot$};
 	
	\node at (0.7, 1.25) {\huge $\cdot$};
	\node at (0.7, 1.85) {\huge $\cdot$};
	\node at (0.7, 2.15) {\huge $\cdot$};
	\node at (0.7, 2.75) {\huge $\cdot$};
	\node at (0.7, 3.35) {\huge $\cdot$};
	\node at (0.7, 3.65) {\huge $\cdot$};
	\node at (0.7, 4.25) {\huge $\cdot$};
	
	\node at (1.3, 1.10) {\huge $\cdot$};
	\node at (1.5, 1.40) {\huge $\cdot$};
	\node at (1.7, 1.10) {\huge $\cdot$};
	
	\node at (2.3, 1.10) {\huge $\cdot$};
	\node at (2.5, 1.40) {\huge $\cdot$};
	\node at (2.7, 1.10) {\huge $\cdot$};
	
	\node at (1.3, 1.85) {\huge $\cdot$};
	\node at (1.5, 2.15) {\huge $\cdot$};
	\node at (1.7, 1.85) {\huge $\cdot$};
	
	\node at (2.3, 1.85) {\huge $\cdot$};
	\node at (2.5, 2.15) {\huge $\cdot$};
	\node at (2.7, 1.85) {\huge $\cdot$};
	
	\node at (1.3, 2.60) {\huge $\cdot$};
	\node at (1.5, 2.90) {\huge $\cdot$};
	\node at (1.7, 2.60) {\huge $\cdot$};
	
	\node at (2.3, 2.60) {\huge $\cdot$};
	\node at (2.5, 2.90) {\huge $\cdot$};
	\node at (2.7, 2.60) {\huge $\cdot$};
	
	\node at (1.3, 3.35) {\huge $\cdot$};
	\node at (1.5, 3.65) {\huge $\cdot$};
	\node at (1.7, 3.35) {\huge $\cdot$};
	
	\node at (2.3, 3.35) {\huge $\cdot$};
	\node at (2.5, 3.65) {\huge $\cdot$};
	\node at (2.7, 3.35) {\huge $\cdot$};
	
	\node at (1.3, 4.10) {\huge $\cdot$};
	\node at (1.5, 4.40) {\huge $\cdot$};
	\node at (1.7, 4.10) {\huge $\cdot$};
	
	\node at (2.3, 4.10) {\huge $\cdot$};
	\node at (2.5, 4.40) {\huge $\cdot$};
	\node at (2.7, 4.10) {\huge $\cdot$};
\end{tikzpicture}
}
\\
\caption{The faces that we have to keep track of are marked with a dot, for states with $n=2$ and $n=3$. The number of dots on a face corresponds to the number of factors $1/\vartheta(\rho\omega^a)$ originating from that face, where $a$ is the height of the face.}
\label{fig:dotpicture}
\end{figure}
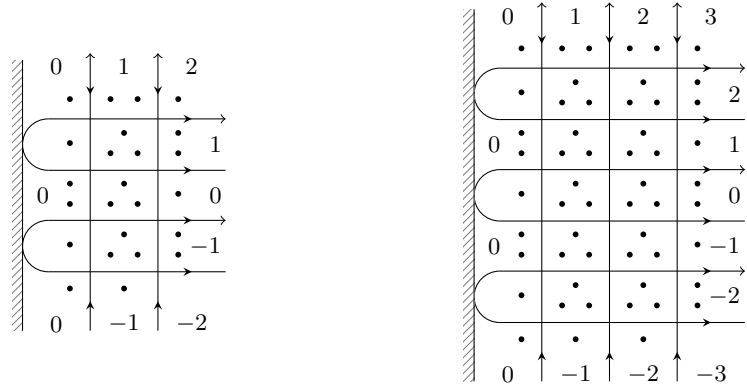

In the next proposition, we will show that we can rewrite the 8VSOS partition function in terms of the partition function of the three-color model.
Define
\begin{gather*}
Z_{n, m}^{3C}(t_0, t_1, t_2)=\sum_{\substack{\text{states with}\\ \text{$m$ positive turns}}} \prod_{\text{faces}} t_i
\end{gather*}
to be the partition function of all three-colorings (with color $0$ fixed in the upper left corner) with a given number of positive turns (i.e., a given number of turns with color $2$), denoted $m$, and where $t_i$ is the weight assigned to color $i$.

\begin{Proposition}\label{myformula}
Let $\eta = -2/3$. Then the partition function is
\begin{gather}
Z_n(\omega, \dots, \omega, 1,\dots, 1, \rho,\zeta) =(-1)^{\binom{n}{2}}\sum_{m=0}^n
\left(\frac{\vartheta\big(\rho \zeta \omega^{-1}\big)}{\vartheta(\rho \zeta \omega)}\right)^m\left(\frac{\vartheta\big(\zeta \omega^{-1}\big)}{\vartheta( \zeta \omega)}\right)^{n-m} \omega^{n^2+n-m}\nonumber
\\ \hphantom{Z_n(\omega, \dots, \omega, 1,\dots, 1, \rho,\zeta)=}
{}\times \vartheta\big(\rho^3, p^3\big)^{2n^2+2n}\vartheta(\rho)^{n+3}\vartheta\big(\rho\omega^{-1}\big)^{2m} \vartheta(\rho\omega)^{2(n-m)}B\nonumber
\\ \hphantom{Z_n(\omega, \dots, \omega, 1,\dots, 1, \rho,\zeta)=}
{}\times Z_{n, m}^{3C}\left(\frac{1}{\vartheta(\rho)^3}, \frac{1}{\vartheta(\rho \omega)^3}, \frac{1}{\vartheta(\rho \omega^2)^3}\right)\!,\label{mposstates}
\end{gather}
with
\begin{gather*}
B=\begin{cases}
1, &\text{for}\quad n\equiv 0 \mod 3,
\\
\dfrac{\vartheta\big(\rho\omega^{-1}\big)}{\vartheta(\rho)},
&\text{for}\quad n\equiv 1 \mod 3,
\\
\dfrac{\vartheta(\rho\omega)\vartheta\big(\rho\omega^{-1}\big)}{\vartheta(\rho)^2},
&\text{for}\quad n\equiv 2 \mod 3.
\end{cases}
\end{gather*}
\end{Proposition}

\begin{proof}
In the partition function \eqref{partfcnabd}, we need to know the heights of three of the adjacent faces to each vertex. Marking the faces that we need to keep track of, as in Fig.~\ref{fig:dotpicture}, we see that each face in the interior of the lattice with height $a$ gives rise to three factors $1/\vartheta(\rho\omega^a)$ in each state in the partition function. The face of each turn generates only one such factor in each state. Elsewhere on the boundary, the number of such factors differs on each face, but there the heights are known (see Fig.~\ref{fig:dotpicture}), so this contribution can be computed explicitly.

Hence the weight of a state can be written in terms of products of $1/\vartheta(\rho\omega^a)^3$ where $a$ is the height of each face, and the correction for the boundaries can be computed explicitly. We~rewrite the partition function \eqref{partfcnabd} as
\begin{gather*}
Z_n(\omega, \dots, \omega, 1,\dots, 1, \rho,\zeta)=C^\prime\sum_{\text{states}} B^\prime \omega^{M}\prod_{\text{faces}}\frac{1}{\vartheta(\rho \omega^a)^3}
\prod_{\text{turns}} \frac{\vartheta\big(\rho^{(1-a)/2} \zeta \omega^{-1}\big)}{\vartheta\big(\rho^{(1-a)/2} \zeta \omega\big)} ,
\end{gather*}
where $C^\prime=(-1)^{\binom{n}{2}}\omega^{n^2}\vartheta\big(\rho^3, p^3\big)^{2n^2}$, and $B^\prime$ is the correction for the boundaries. To compute~$B'$, first realize that the faces of the turns are accounted for two times too much, which yields the correction
\begin{gather*}
\prod_{\text{turns}}\vartheta(\rho\omega^a)^2.
\end{gather*} Along the rest of the boundary at the reflecting wall, the $0$-faces are counted $n+3$ times too much, which gives the factor $\vartheta(\rho)^{n+3}$ to the correction. Let $B^{\prime\prime}$ be the joint correction on the remaining three boundaries. This part depends on the value of~$n$ modulo $3$. The correction is
\begin{gather*}
B^{\prime\prime}=
\begin{cases}
\vartheta\big(\rho^3, p^3\big)^{2n}, &\text{for}\quad n\equiv 0\mod 3,
\\
\vartheta(\rho)^{-1}\vartheta\big(\rho\omega^{-1}\big)\vartheta\big(\rho^3, p^3\big)^{2n},
&\text{for}\quad n\equiv 1\mod 3,
\\
\vartheta(\rho)^{-3}\vartheta\big(\rho^3, p^3\big)^{2n+1}, &\text{for}\quad n\equiv 2\mod 3,
\end{cases}
\end{gather*}
where we used \eqref{mod3} and \eqref{threeproduct} to simplify the expressions.
Putting everything together yields
\begin{gather*}
Z_n(\omega, \dots, \omega, 1,\dots, 1, \rho,\zeta)
=(-1)^{\binom{n}{2}}\omega^{n^2} \vartheta\big(\rho^3, p^3\big)^{2n^2+2n}\vartheta(\rho)^{n+3}B
\\ \hphantom{Z_n(\omega, \dots, \omega, 1,\dots, 1, \rho,\zeta)=}
{}\times\sum_{\text{states}} \omega^{M} \prod_{\text{faces}}\frac{1}{\vartheta(\rho \omega^a)^3}
\prod_{\text{turns}}\left(\vartheta(\rho\omega^a)^2\frac{\vartheta\big(\rho^{(1-a)/2} \zeta \omega^{-1}\big)}{\vartheta\big(\rho^{(1-a)/2} \zeta \omega\big)}\right)\!,
\end{gather*}
with
\begin{gather*}
B=\begin{cases}
1, &\text{for}\quad n\equiv 0 \mod 3,
\\
\dfrac{\vartheta\big(\rho\omega^{-1}\big)}{\vartheta(\rho)}, &\text{for}\quad n\equiv 1 \mod 3,\vspace{.5ex}
\\
\dfrac{\vartheta(\rho\omega)\vartheta\big(\rho\omega^{-1}\big)}{\vartheta(\rho)^2},
&\text{for}\quad n\equiv 2 \mod 3.
\end{cases}
\end{gather*}
We rewrite the partition function as a sum over the number of positive turns in each state,
which gives~\eqref{mposstates}.
\end{proof}

\section{Rewriting of Filali's determinant formula}
\label{sec4}
In this section, we rewrite Filali's determinant formula as we did with the partition function in Section~\ref{sec3}. We~do the same variable changes as in Section~\ref{sectionchangepartitionfunction} and specify the parameter $\eta=-2/3$. Before we specialize the values $\lambda_i=-1/2$ and $\mu_j=0$, we rewrite the determinant in~terms of Bazhanov's and Mangazeev's polynomials $q_n$. Then we rewrite the partition function as a sum to be able to compare the terms pairwise with the terms in \eqref{mposstates}. In this way, we get an expression for the partition function of the three-color model in terms of~$q_n$.

\subsection[Filali's determinant in terms of T(2psi+1, ..., 2psi+1)]{Filali's determinant in terms of $\boldsymbol{T(2\psi+1, \dots, 2\psi+1)}$}
\label{secD}
Before we can specialize $\lambda_i=-1/2$ and $\mu_j=0$, we need to rewrite Filali's determinant.
To be able to do this, we define \cite{Rosengren2015}
\begin{gather*}
\psi:=\psi(\tau)=\frac{\omega^2\vartheta(-1)\vartheta\big({-}p^{1/2}\omega\big)}{\vartheta\big({-}p^{1/2}\big)\vartheta(-\omega)},
\qquad
x(z)=\frac{\vartheta\big({-}p^{1/2}\omega\big)^2\vartheta\big(\omega {\rm e}^{\pm 2\pi {\rm i} z}\big)}{\vartheta(-\omega)^2 \vartheta\big(p^{1/2} \omega {\rm e}^{\pm 2\pi {\rm i} z}\big)}
\end{gather*}
(in~\cite{Rosengren2015} $\psi$ is denoted $\zeta$), and
\begin{gather}
\label{T}
T(x_1, \dots, x_{2n})=\frac{\prod\limits_{i,j=1}^n G(x_j,x_{n+i})}{\Delta(x_1, \dots, x_n)\Delta(x_{n+1}, \dots, x_{2n})}\det_{1\leq i,j\leq n}\left(\frac{1}{G(x_j, x_{n+i})}\right)\!,
\end{gather}
where $\Delta(x_1, \dots, x_n)=\prod_{1\leq i<j\leq n} (x_j-x_i)$, and
\begin{gather*}
G(x,y)=(\psi+2)xy(x+y)+\psi(2\psi+1)(x+y)-2\big(\psi^2+3\psi+1\big)xy-\psi\big(x^2+y^2\big).
\end{gather*}
$T$ is a symmetric polynomial \cite{Rosengren2014-1}.
For $\psi$, the following identities hold \cite[Lemma~9.1]{Rosengren2011}:
\begin{gather}
2\psi+1=\frac{\vartheta\big({-}p^{1/2}\omega\big)^2\vartheta(\omega)^2}{\vartheta(-\omega)^2\vartheta\big(p^{1/2}\omega\big)^2}, \label{lemma9.1-2psi+1}
\\
\psi+1=-\frac{\vartheta\big(p^{1/2}\big)\vartheta\big({-}p^{1/2}\omega\big)}{\vartheta\big({-}p^{1/2}\big)\vartheta\big(p^{1/2}\omega\big)},
\label{lemma9.1-psi+1}
\\
\psi-1=\frac{\vartheta\big(p^{1/2}\big)\vartheta\big(p^{1/2}\omega\big)\vartheta(\omega)^2}
{\vartheta(-p^{1/2})\vartheta(-p^{1/2}\omega)\vartheta(-\omega)^2}.\label{lemma9.1-psi-1}
\end{gather}
Another useful identity, which follows from the addition rule
\eqref{additionrule} and \eqref{omegaomega2}, is
\begin{gather}
\label{xzminusxw}
x(z)-x(w)
=\frac{\vartheta\big({-}p^{1/2}\omega\big)^2\vartheta\big(p^{1/2}\omega\big)\vartheta\big(p^{1/2}\big)\omega}{\vartheta(-\omega)^2}
\frac{{\rm e}^{-2\pi {\rm i} w}\vartheta\big({\rm e}^{2\pi{\rm i}(w\pm z)}\big)}{\vartheta\big(p^{1/2} \omega {\rm e}^{\pm 2\pi {\rm i} z}\big)\vartheta\big(p^{1/2} \omega {\rm e}^{\pm 2\pi {\rm i} w}\big)}.
\end{gather}

Now consider
\begin{gather*}
\tilde D=\frac{\prod\limits_{i,j=1}^n \vartheta\big(q^{\lambda_i+\mu_j+1}\big)\vartheta\big(q^{\lambda_i-\mu_j+1}\big) \vartheta\big(q^{\lambda_i+\mu_j}\big)\vartheta\big(q^{\lambda_i-\mu_j}\big)}{\prod\limits_{1\leq i<j\leq n} q^{-\lambda_i-\mu_j}\vartheta\big(q^{\lambda_i+\lambda_j+1}\big)\vartheta\big(q^{\lambda_i-\lambda_j}\big)
\vartheta\big(q^{\mu_j+\mu_i}\big)\vartheta\big(q^{\mu_j-\mu_i}\big)} \det_{1\leq i,j \leq n} K_{ij}.
\end{gather*}
Put $z_{n+i}=-2(\lambda_i+1/2)/3$ and $z_j=\mu_j/3$ for $1\leq i,j\leq n$. Then $\lambda_i=-1/2$ and $\mu_j=0$ correspond to $z_i=0$, for all $i$. Using \eqref{threeproduct}, we can write $K_{ij}$ as
\begin{gather*}
K_{ij}=\frac{\vartheta\big({\rm e}^{2\pi{\rm i}(z_{n+i}\pm z_j)}\big)}
{\vartheta\big({\rm e}^{6\pi {\rm i}(z_{n+i}\pm z_j)}, p^3\big)}.
\end{gather*}
We want to rewrite $K_{ij}$ using the following lemma. The equation can be found in~\cite{Rosengren2014-1}, although the constant is not written out explicitly there.
\begin{Lemma}
We have
\begin{gather*}
\frac{\vartheta\big({\rm e}^{2\pi{\rm i} (w\pm z)}\big)}{\vartheta\big({\rm e}^{6\pi {\rm i} (w\pm z)}, p^3\big)}
=\frac{\tilde C {\rm e}^{-4\pi {\rm i} w}}{\vartheta\big(p^{1/2}\omega {\rm e}^{\pm 2\pi {\rm i} w}\big)^2\vartheta\big(p^{1/2}\omega {\rm e}^{\pm 2\pi {\rm i} z}\big)^2} \frac{1}{G(x(z),x(w))},
\end{gather*}
with
\begin{gather*}
\tilde C=\frac{\omega^2\vartheta(-1)\vartheta\big(p^{1/2}\big)^3\vartheta\big(p^{1/2}\omega\big)^2\vartheta\big({-}p^{1/2}\omega\big)^6} {\vartheta(-\omega)^4\vartheta\big({-}p^{1/2}\big)}.
\end{gather*}
\end{Lemma}

\begin{proof}
Putting equations (2.17) (observe the misprint, a factor ${\rm e}^{-2\pi {\rm i} z}$ is missing in the numerator on the right hand side) and (2.23) of~\cite{Rosengren2014-1} together, yields
\begin{gather*}
\frac{\vartheta\big({\rm e}^{2\pi{\rm i} (w\pm z)}\big)}{\vartheta\big({\rm e}^{6\pi {\rm i} (w\pm z)}, p^3\big)}
=\frac{\tilde C {\rm e}^{-4\pi {\rm i} w}}{\vartheta\big(p^{1/2}\omega {\rm e}^{\pm 2\pi {\rm i} w}\big)^2\vartheta\big(p^{1/2}\omega {\rm e}^{\pm 2\pi {\rm i} z}\big)^2} \frac{1}{G(x(z),x(w))}.
\end{gather*}
To get the constant, put $z=w=1/2$. It~is easy to see that $x(1/2)=1$. Now
\begin{gather*}
\tilde C =\frac{\vartheta\big({-}p^{1/2}\omega\big)^8}{\omega^4\vartheta(\omega)^4}G(1, 1),
\end{gather*}
and, using \eqref{Rosengren4.2a} and \eqref{lemma9.1-psi-1},
\begin{gather*}
G(1, 1)=2(\psi-1)^2=\frac{\vartheta(-1)\vartheta\big(p^{1/2}\big)^3\vartheta\big(p^{1/2}\omega\big)^2 \vartheta(\omega)^4}{\vartheta\big({-}p^{1/2}\big)\vartheta\big({-}p^{1/2}\omega\big)^2\vartheta(-\omega)^4},
\end{gather*}
which together yield $\tilde C$ as stated in the lemma.
\end{proof}

Using the above lemma and \eqref{xzminusxw} yields
\begin{gather*}
\tilde D
=(-1)^{\binom{n}{2}}\left(\frac{\vartheta(-\omega)^2\vartheta\big({-}p^{1/2}\big)}
{\omega^2\vartheta(-1)\vartheta\big(p^{1/2}\big)^2\vartheta\big(p^{1/2}\omega\big)\vartheta\big({-}p^{1/2}\omega\big)^4}\right)^{n(n-1)}
\nonumber\\ \hphantom{\tilde D=}
{}\times\prod_{i=1}^n \big(\big({\rm e}^{4\pi {\rm i} z_{n+i}}\big)^{n-1} \vartheta\big(p^{1/2}\omega {\rm e}^{\pm 2\pi {\rm i} z_i}\big)^{n-1}\vartheta\big(p^{1/2}\omega {\rm e}^{\pm 2\pi {\rm i} z_{n+i}}\big)^{n-1}\big) T(x(z_1), \dots, x(z_{2n})).
\end{gather*}
Now we can put $z_i=0$ in $\tilde D$. It~is easy to see that $x(0)=2\psi+1$, by using \eqref{lemma9.1-2psi+1}. We~get
\begin{gather*}
\tilde D
=(-1)^{\binom{n}{2}}\left(\frac{\vartheta(-\omega)^2\vartheta\big({-}p^{1/2}\big)\vartheta\big(p^{1/2}\omega\big)^3}
{\omega^2\vartheta(-1)\vartheta\big(p^{1/2}\big)^2\vartheta\big({-}p^{1/2}\omega\big)^4}\right)^{n(n-1)} T(2\psi+1, \dots, 2\psi+1).
\end{gather*}
In \cite{Rosengren2015}, $T(2\psi+1, \dots, 2\psi+1)$ is denoted $t^{(2n,0,0,0)}(\psi)$.

\subsection[Filali's determinant in terms of the polynomials qn]{Filali's determinant in terms of the polynomials $\boldsymbol{q_n}$}
\label{secqn}
In \cite{BazhanovMangazeev2005}, Bazhanov and Mangazeev found certain polynomials
\begin{gather*}
\mathcal{P}_n(x, z)=\sum_{k=0}^n r_k^{(n)}(z) x^k,
\end{gather*}
normalized by $r_n^{(n)}(0)=1$, which describe the ground state eigenvalue of Baxter's $Q$-operator \cite{Baxter1972} for the 8V model in the case with $\eta=-2/3$. They also introduced polynomials $s_n(z)=r_n^{(n)}(z)$ and $\overline{s}_n(z)=r_0^{(n)}(z)$. In~\cite{BazhanovMangazeev2010} they connect these polynomials to the ground state eigenvectors of the supersymmetric XYZ-Hamiltonian for spin chains of odd length $2n+1$. They state several conjectures about these polynomials, among them that $s_n(z)$ can be factorized into polynomials which seem to have positive coefficients. Here certain polynomials $q_{n-1}(z)$, with $\deg q_n(z)=n(n+1)$ and $q_n(0)=1$, show up as factors of~$s_{2n}\big(z^2\big)$.
Other conjectures include that some components of the ground state eigenvectors for the XYZ spin chain can be written in terms of~$s_n(z)$, $\overline{s}_n(z)$ and $q_n(z)$. The polynomials $q_n(z)$ have the symmetries \cite{BazhanovMangazeev2010}
\begin{gather}
\label{symmetries}
q_n(z)=q_n(-z) \qquad \text{and} \qquad q_n(z)=\left(\frac{1+3z}{2}\right)^{n(n+1)} q_n\left(\frac{1-z}{1+3z}\right)\!.
\end{gather}
Zinn-Justin~\cite{Zinn-Justin2013} observed that the polynomials $q_n(z)$ seem to be given by specializing the variables in a determinant equivalent to \eqref{T}. Using \cite[equation~(5.5) and Proposition 2.2]{Rosengren2015}, this identity takes the form
\begin{gather}
q_{n-1}\left(\frac{1}{2\psi+1}\right)
=\left(\frac{1}{(\psi+1)(2\psi+1)^2}\right)^{n(n-1)}T(2\psi+1, \dots, 2\psi+1).\label{qnequalst2n}
\end{gather}
In the present work we take \eqref{qnequalst2n} as the definition of~$q_{n-1}$. The identification of these polynomials with the ones introduced in~\cite{BazhanovMangazeev2010} should still be viewed as a conjecture. We~put \eqref{qnequalst2n} into $\tilde D$ and~get
\begin{gather*}
\tilde D=(-1)^{\binom{n}{2}}\!\left(\!\frac{\vartheta(-\omega)^2\vartheta\big({-}p^{1/2}\big)\vartheta\big(p^{1/2}\omega\big)^3}
{\omega^2\vartheta(-1)\vartheta\big(p^{1/2}\big)^2\vartheta\big({-}p^{1/2}\omega\big)^4}\!\right)^{n(n-1)} \!\!\!\!\!
\big((\psi+1)(2\psi+1)^2\big)^{n(n-1)} q_{n-1}\!\left(\frac{1}{2\psi+1}\right)\!.
\end{gather*}

\subsection{Filali's determinant formula as a sum}
Now we specify $\lambda_i=-1/2$, $\mu_i=0$ and $\eta=-2/3$ in Filali's determinant formula \eqref{Filalisdeterminantformula}.
Recall that $\omega=q^{-1/2}=q$. We~get
\begin{gather*}
Z_n(\omega, \dots, \omega, 1,\dots, 1, \rho,\zeta)=(-1)^n\vartheta(\omega)^{2(n-n^2)}\left(\frac{\vartheta(\zeta)\vartheta(\rho \zeta)}{\vartheta(\zeta \omega)\vartheta(\rho\zeta \omega)}\right)^n \tilde B\tilde D,
\end{gather*}
where
\begin{gather*}
\tilde B=\prod_{i=1}^n\frac{\vartheta\big(\rho \omega^{2i-n-2}\big)}{\vartheta(\rho \omega^{n-i})}=
\begin{cases}
1, &\text{for}\quad n\equiv 0, 2 \mod 3,
\\
\dfrac{\vartheta\big(\rho \omega^{-1}\big)}{\vartheta(\rho)}, &\text{for}\quad n\equiv 1 \mod 3.
\end{cases}
\end{gather*}
In analogy with \eqref{mposstates}, we want to write Filali's determinant formula as a sum over $m$.
First we write $\left(\frac{\vartheta(\zeta)\vartheta(\rho \zeta)}{\vartheta(\zeta \omega)\vartheta(\rho\zeta \omega)}\right)^n$ in terms of~$\frac{\vartheta(\rho \zeta \omega^{-1})}{\vartheta(\rho \zeta \omega)}$ and $\frac{\vartheta(\zeta \omega^{-1})}{\vartheta( \zeta \omega)}$. For $n=1$, we want to solve
\begin{gather*}
\frac{\vartheta(\zeta)\vartheta(\rho \zeta)}{\vartheta(\zeta \omega)\vartheta(\rho\zeta \omega)}=P_1\frac{\vartheta\big(\rho \zeta \omega^{-1}\big)}{\vartheta(\rho \zeta \omega)}+ P_2 \frac{\vartheta\big(\zeta \omega^{-1}\big)}{\vartheta( \zeta \omega)}.
\end{gather*}
Letting $\zeta=\omega$ and $\zeta=\omega^{-1}$ respectively yields
\begin{gather*}
P_1=-\frac{\omega\vartheta(\rho \omega)}{\vartheta(\rho)}, \qquad \text{and} \qquad
P_2 =-\frac{\omega^2\vartheta\big(\rho \omega^2\big)}{\vartheta(\rho)}.
\end{gather*}
The addition rule \eqref{additionrule} assures that this is a solution. We~put this into the partition function for general $n$, and use the binomial theorem to obtain
\begin{gather*}
Z_n(\omega, \dots, \omega, 1,\dots, 1, \rho,\zeta)
\\ \qquad{}
=(-1)^n\vartheta(\omega)^{2(n-n^2)}\left(\!-\frac{\omega\vartheta(\rho \omega)}{\vartheta(\rho)}\frac{\vartheta\big(\rho \zeta \omega^{-1}\big)}{\vartheta(\rho \zeta \omega)} -\frac{\omega^2\vartheta\big(\rho \omega^2\big)}{\vartheta(\rho)} \frac{\vartheta\big(\zeta \omega^{-1}\big)}{\vartheta( \zeta \omega)}\right)^n \tilde B\tilde D
\\ \qquad{}
= \sum_{m=0}^n \! \binom{n}{m} \omega^{2n-m}\vartheta(\omega)^{2(n-n^2)}\frac{\vartheta(\rho\omega)^m \vartheta\big(\rho\omega^{2}\big)^{n-m}}{\vartheta(\rho)^n}\!\left(\!\frac{\vartheta\big(\rho \zeta \omega^{-1}\big)}{\vartheta(\rho \zeta \omega)}\!\right)^m\!\!\left(\!\frac{\vartheta\big(\zeta \omega^{-1}\big)}{\vartheta(\zeta \omega)}\!\right)^{n-m}\!\!\tilde B\tilde D.
\end{gather*}
Finally, inserting the expression for $\tilde D$ yields
\begin{gather*}
Z_n(\omega, \dots, \omega, 1, \dots, 1, \rho,\zeta)
\\ \qquad{}
=(-1)^{\binom{n}{2}}\sum_{m=0}^n\left(\frac{\vartheta\big(\rho \zeta \omega^{2}\big)}{\vartheta(\rho \zeta \omega)}\right)^m\left(\frac{\vartheta\big(\zeta \omega^{-1}\big)}{\vartheta(\zeta \omega)}\right)^{n-m}\omega^{n^2+n-m}
\\ \qquad\phantom{=}
\times\binom{n}{m} \frac{\vartheta(\rho\omega)^m \vartheta\big(\rho\omega^{-1}\big)^{n-m}}{\vartheta(\rho)^n} \left(\frac{\vartheta(-\omega)^2\vartheta\big({-}p^{1/2}\big)\vartheta\big(p^{1/2}\omega\big)^3}
{\vartheta(\omega)^2\vartheta(-1)\vartheta\big(p^{1/2}\big)^2\vartheta\big({-}p^{1/2}\omega\big)^4}\right)^{n(n-1)}
\\ \qquad\phantom{=}
\times\tilde B ((\psi+1)(2\psi+1)^2)^{n(n-1)} q_{n-1}\left(\frac{1}{2\psi+1}\right)\!,
\end{gather*}
where
\begin{gather*}
\tilde B=
\begin{cases}
1, &\text{for} \quad n\equiv 0, 2 \mod 3,
\\
\dfrac{\vartheta\big(\rho \omega^{-1}\big)}{\vartheta(\rho)}, &\text{for} \quad n\equiv 1 \mod 3.
\end{cases}
\end{gather*}

Now we can compare this with \eqref{mposstates}. The terms with different $m$ are linearly independent as functions of~$\zeta$. This follows since the $m$th term has a zero of degree $n-m$ in $\zeta=\omega$. Therefore we can identify the terms with the same $m$. We~get the following expression for the partition function.
\begin{Lemma}
The partition function of the three-color model for a fixed $m$ is
\begin{gather}
Z_{n, m}^{3C}\left(\frac{1}{\vartheta(\rho)^3}, \frac{1}{\vartheta(\rho \omega)^3}, \frac{1}{\vartheta\big(\rho \omega^2\big)^3}\right) \nonumber
\\ \qquad{}
=\binom{n}{m} \left(\frac{\vartheta(-\omega)^2\vartheta\big({-}p^{1/2}\big)\vartheta\big(p^{1/2}\omega\big)^3}
{\vartheta(\omega)^2\vartheta(-1)\vartheta\big(p^{1/2}\big)^2\vartheta\big({-}p^{1/2}\omega\big)^4}\right)^{n(n-1)}\nonumber
\\ \qquad\phantom{=}
{}\times\frac{\hat B ((\psi+1)(2\psi+1)^2)^{n(n-1)}}{\vartheta(\rho)^{2n^2+4n+3} \vartheta(\rho\omega)^{2n^2+4n-3m} \vartheta\big(\rho\omega^2\big)^{2n^2+n+3m}} q_{n-1}\left(\frac{1}{2\psi+1}\right)\!,
\label{neweqMeFilali}
\end{gather}
where
\begin{gather*}
\hat B=\begin{cases}
1, &\text{for}\quad n\equiv 0, 1 \mod 3,
\\
\dfrac{\vartheta(\rho)^2}{\vartheta(\rho\omega)\vartheta\big(\rho\omega^2\big)},
&\text{for}\quad n\equiv 2 \mod 3.
\end{cases}
\end{gather*}
\end{Lemma}

\section{The main result}\label{secmainresult}
In this section, we rewrite \eqref{neweqMeFilali} in algebraic form. We~will need the following identities:
\begin{gather}
\label{eq:n02}
\frac{\vartheta(-\omega)^3}{\vartheta(-1)^3}=\frac{\psi+1}{2\psi^2},
\\
\label{eq:n01}
\left(\frac{\vartheta(\omega)^2\vartheta(-1)\vartheta\big(p^{1/2}\big)^2\vartheta\big({-}p^{1/2}\omega\big)^4}{\vartheta(-\omega)^2 \vartheta\big({-}p^{1/2}\big)\vartheta\big(p^{1/2}\omega\big)^3}\right)^6
=2^4\psi^{2}(\psi+1)^8(2\psi+1)^6,
\end{gather}
and
\begin{gather}
\label{eq:n03}
\omega^4\left(\frac{\vartheta(-1)}{\vartheta(-\omega)}\right)^2\left(\frac{\vartheta(\omega)^2 \vartheta(-1)\vartheta\big(p^{1/2}\big)^2\vartheta\big({-}p^{1/2}\omega\big)^4}{\vartheta(-\omega)^2\vartheta\big({-}p^{1/2}\big) \vartheta\big(p^{1/2}\omega\big)^3}\right)^2
=\left(2 \psi(\psi+1)(2\psi+1)\right)^2.
\end{gather}
The above identities can be found using \cite[Lemmas~3.1 and~3.5]{Rosengren2014-3}. Once we have the expressions, it is much easier to go in the other direction, from the expressions in terms of~$\psi$ to~the~theta functions, by using the definition of~$\psi$ and equations \eqref{Rosengren4.2a}, \eqref{Rosengren4.2b}, \eqref{lemma9.1-2psi+1} and \eqref{lemma9.1-psi+1}.

Introduce new variables
$t_i=1/\vartheta(\rho \omega^i)^3$ and define
\begin{gather*}
T=\frac{(t_0t_1+t_0t_2+t_1t_2)^3}{(t_0t_1t_2)^2}.
\end{gather*}
We want to rewrite the partition function as a polynomial in $T$. We~will need \cite[Lemmas~5.1 and~5.3]{Rosengren2011}, which we state here without proof.

\begin{Lemma}\label{Lemma5.1Rosengren}
There exists a function $p \mapsto f(p)$, which does not depend on $\rho$, such that
\begin{gather*}
\vartheta\big(\rho^3, p^3\big)f(p)=\frac{1}{t_0}+\frac{1}{t_1}+\frac{1}{t_2}.
\end{gather*}
Moreover $T=f(p)^3$.
\end{Lemma}
Since $f(p)$ is independent of~$\rho$, we can put $\rho=-1$ to get the expression
\begin{gather}\label{fp}
f(p)=\frac{2\vartheta(-\omega)^3/\vartheta(-1)^3+1}{\omega^2\vartheta(-\omega)^2/\vartheta(-1)^2}.
\end{gather}
It follows that $T$ is also independent of~$\rho$, and \eqref{eq:n02} yields that
\begin{gather*}
T=\frac{4\big(\psi^2+\psi+1\big)^3}{\psi^2(\psi+1)^2}.
\end{gather*}

\begin{Lemma}\label{Lemma5.3Rosengren}Let $f$ be a Laurent polynomial in three variables $t_0$, $t_1$ and $t_2$, homogeneous of degree $0$. Suppose that under the parametrization $t_i=1/\vartheta(\rho \omega^i)^3,$ the polynomial $f(t_0, t_1, t_2)$ is independent of~$\rho$. Then $f$ is a polynomial in~$T$.
\end{Lemma}

The following theorem is our main result, equivalent to Theorem~\ref{maintheorem}.
\begin{Theorem}
\label{proppT}
It holds that
\begin{gather*}
Z_{n, m}^{3C}(t_0, t_1, t_2)\\ \qquad
{}=\begin{cases}
\displaystyle\binom{n}{m} \frac{t_2^{m-n}}{t_1^m}t_0(t_0t_1t_2)^{(2n^2+4n)/3} Q(T), &\textup{$n\equiv 0, 1 \mod 3$},\vspace{.5ex}
\\
\displaystyle\binom{n}{m} \frac{t_2^{m-n}}{t_1^m}(t_0t_1+t_0t_2+t_1t_2)(t_0t_1t_2)^{(2n^2+4n-1)/3}Q(T), &\textup{$n\equiv 2 \mod 3$},
\end{cases}
\end{gather*}
where $Q(T)$ is a polynomial in $T$ with
\begin{gather*}
Q(T)=
\begin{cases}
\dfrac{q_{n-1}(z)}{(z(z^2-1))^{(n^2-n)/3}} , &\text{for}\quad n\equiv 0, 1 \mod 3,\\
\dfrac{q_{n-1}(z)}{(3z^2+1)(z(z^2-1))^{(n^2-n-2)/3}}, &\text{for}\quad n\equiv 2 \mod 3,
\end{cases}
\end{gather*}
for
\begin{gather*}
T=\frac{(3z^2+1)^3}{(z(z^2-1))^2}.
\end{gather*}
\end{Theorem}

\begin{proof}
First consider $n\equiv 0, 1 \mod 3$. In \eqref{neweqMeFilali}, change to the variables $t_i$. The partition function becomes
\begin{gather}
Z_{n, m}^{3C}(t_0, t_1, t_2)
=A\binom{n}{m} \frac{t_2^{m-n}}{t_1^m}t_0(t_0t_1t_2)^{(2n^2+4n)/3},\label{n01proofmain}
\end{gather}
where
\begin{gather*}
\begin{split}
A=&\left(\frac{\vartheta(-\omega)^2\vartheta\big({-}p^{1/2}\big)\vartheta\big(p^{1/2}\omega\big)^3}{\vartheta(\omega)^2 \vartheta(-1)\vartheta\big(p^{1/2}\big)^2\vartheta\big({-}p^{1/2}\omega\big)^4}\right)^{n(n-1)} \big((\psi+1)(2\psi+1)^2\big)^{n(n-1)}q_{n-1}\left(\frac{1}{2\psi+1}\right)\!.
\end{split}
\end{gather*}
Observe that $A$ does not depend on $\rho$ or $m$. We~will show that $A$ can be written as a polynomial in $T$.

Since $2n^2+4n\equiv 0 \mod 3$, the exponents are integers. As a polynomial in $t_0$, $t_1$ and $t_2$, the left hand side of~\eqref{n01proofmain} is homogenous of degree $2n^2+3n+1$, and the degree of~$t_0$, $t_1$, $t_2$ on the right hand side adds up to $2n^2+3n+1$ as well, so
\begin{gather*}
A=\frac{t_1^mt_2^{n-m}}{\binom{n}{m}(t_0t_1t_2)^{(2n^2+4n)/3}t_0}Z_{n, m}^{3C}(t_0, t_1, t_2)
\end{gather*}
is a Laurent polynomial of degree $0$.
Using Lemma~\ref{Lemma5.3Rosengren}, we get that $A$ is a polynomial in $T$.

For $n \equiv 2 \mod 3$, the partition function \eqref{neweqMeFilali} is
\begin{gather*}
Z_{n, m}^{3C}(t_0, t_1, t_2)=A\binom{n}{m}
\frac{t_2^{m-n}}{t_1^m}\left(\frac{1}{\vartheta(\rho)\vartheta(\rho\omega)\vartheta(\rho\omega^2)}\right)^{2n^2+4n+1}.
\end{gather*}
Since $2n^2+4n+1\equiv 2 \mod 3$, we instead look at
\begin{gather*}
X=\frac{(t_0t_1+t_0t_2+t_1t_2)^2 t_1^m t_2^{n-m}}{\binom{n}{m}(t_0t_1t_2)^{(2n^2+4n+5)/3}} Z_{n, m}^{3C}(t_0, t_1, t_2),
\end{gather*}
which is a Laurent polynomial of degree $0$.
To see that $X$ is independent of~$\rho$, observe that $X=A(f(p))^2$ where $f(p)$ is the function from Lemma~\ref{Lemma5.1Rosengren}. Both $A$ and $f(p)$ are independent of~$\rho$, so $X$ is as well. Using Lemma~\ref{Lemma5.3Rosengren} we can conclude that $X$ is a polynomial in $T$. We~see that $X=0$ whenever $T=0$, so $X$ is divisible by $T$. Hence $X=T Q(T)$ for some polynomial~$Q(T)$. Thus
\begin{gather*}
Z_{n, m}^{3C}(t_0, t_1, t_2)
=\binom{n}{m}\frac{t_2^{m-n}}{t_1^m}(t_0t_1+t_0t_2+t_1t_2)(t_0t_1t_2)^{(2n^2+4n-1)/3} Q(T).
\end{gather*}
Hence for all $n$, the partition function can be written in terms of a polynomial $Q(T)$.

Now we will compute the polynomial $Q(T)$. For $n\equiv 0,1 \mod 3$, we have
$n(n-1)\equiv 0 \mod 6$, and we know that $Q(T)=A$, so by using \eqref{eq:n01} we get
\begin{gather*}
Q(T)=\frac{(2\psi+1)^{n(n-1)}}{(4\psi(\psi+1))^{\frac{n(n-1)}{3}}}q_{n-1}\left(\frac{1}{2\psi+1}\right)\!.
\end{gather*}
For $n\equiv 2 \mod 3$, we have $Q(T)=A/f(p)$. Since $n(n-1)\equiv 2 \mod 6$,
we need to be careful when changing to the variables $t_i$. Using \eqref{eq:n02}--\eqref{eq:n03}, and inserting \eqref{fp}, we get
\begin{gather*}
Q(T)=\frac{(2\psi+1)^{n(n-1)}}{4\left(\psi^2+\psi+1\right)(4\psi(\psi+1))^{(n^2-n-2)/3}}q_{n-1}\left(\frac{1}{2\psi+1}\right)\!.
\end{gather*}
Changing to the variable $z=\frac{1}{2\psi+1}$ yields the desired result.
\end{proof}

\section{Consequences of Theorem~\ref{maintheorem}}\label{secconsequences}

From Theorem \ref{maintheorem} we can derive several consequences. For instance, we get an explicit formula for the polynomials $q_n(z)$ and we can prove that $q_n(z+1)$ and $(z+1)^{n(n+1)}q_n(1/(z+1))$ have positive coefficients. Furthermore we can compute strict bounds for the number of faces of each color in the three-colorings.

\subsection[Consequences for the polynomials qn]{Consequences for the polynomials $\boldsymbol{q_n}$}
Rearranging Theorem~\ref{maintheorem} yields
\begin{gather*}
q_{n-1}(z)=
\begin{cases}
\displaystyle\frac{t_1^m t_2^{n-m}(z(z^2-1))^{\frac{n^2-n}{3}}}{\binom{n}{m}t_0(t_0t_1t_2)^{\frac{2n^2+4n}{3}}}Z_{n, m}^{3C}(t_0, t_1, t_2),
& n\equiv 0, 1 \mod 3,
\\
\displaystyle\frac{t_1^m t_2^{n-m}(3z^2+1)(z(z^2-1))^{\frac{n^2-n-2}{3}}}{\binom{n}{m}(t_0t_1+t_0t_2+t_1t_2)(t_0t_1t_2)^{\frac{2n^2+4n-1}{3}}}Z_{n, m}^{3C}(t_0, t_1, t_2),
&n\equiv 2 \mod 3.
\end{cases}
\end{gather*}
For instance, for $t_0=z(z+1)/(z-1)^2$, $t_1=t_2=1$, and $m=0$, we get
\begin{gather}\label{b}
q_{n-1}(z)=\sum_{k_0\in \Z}N^{(0)}(k_0)(z(z+1))^{k_0-(n^2+5n+a)/3}(z-1)^{(5n^2+7n+2a)/3-2k_0},
\end{gather}
with
\begin{gather*}
a=\begin{cases}
3, &\text{for}\quad n\equiv 0, 1 \mod 3,\\
1, &\text{for}\quad n\equiv 2 \mod 3,
\end{cases}
\end{gather*}
and where $N^{(m)}(k_0)$ denotes the number of states with exactly $m$ positive turns, and $k_0$ faces of color $0$.

The coefficients of~$q_n(z)$ are all integers. It~has been conjectured that $q_n(z)$ has only positive integer coefficients \cite{BazhanovMangazeev2010}. This is not clear from the expression \eqref{b}. It~is not enough to notice that $N^{(0)}(k_0)$ is always non-negative, one would need some further constraints on $N^{(0)}(k_0)$. However, we have the following weaker result.

\begin{Corollary}\label{cor:poscoeffs}
The polynomials $(z+1)^{n(n+1)} q_n\big(\frac{1}{z+1}\big)$ and $q_n(z+1)$ have positive integer coefficients.
\end{Corollary}

\begin{proof}Put $z+1$ and $1/(z+1)$ respectively into \eqref{b}. For $z+1$ the statement is clear. For $1/(z+1)$ it is enough to notice that $(5n^2+7n+2a)/3-2k_0$ is an even number for all $n$.
\end{proof}

As is explained in Section~2.9 of~\cite{Rosengren2014-1}, $\psi$ is a Hauptmodul for $\Gamma_0(12)$. This means that $\psi$ generates the corresponding field of modular functions. The variable $z=1/(2\psi+1)$ is another Hauptmodul. The six cusps of~$\Gamma_0(12)$ are at
\begin{gather*}
z=0,\ -1,\ 1,\ -1/3,\ 1/3,\ \infty.
\end{gather*}
It is natural to consider Hauptmodulen of the form $(z-\alpha)/(z-\beta)$, where $\alpha$ and $\beta$ are one of the cusps ($\infty$ interpreted as a limit). Because of the symmetries \eqref{symmetries} of~$q_n$, the different variables generate only four essentially different polynomials, up to scaling the variable or replacing it~by~its inverse. The essentially different variables are $z$, the two variables in Corollary~\ref{cor:poscoeffs}, and one more, given by $\upsilon=3z-1$. For $q_n((\upsilon+1)/3)$ it is also not directly clear that the coefficients are positive.

\subsection{Consequences for the three-color model}

In each state of the three-color model, let $k_i$ be the number of faces with the $i$th color and let $N^{(m)}(k_0, k_1, k_2)$ denote the number of states with exactly $m$ positive turns and $k_i$ entries of color~$i$.
Now the partition function can be written
\begin{gather}\label{generatingfunction}
\sum_{(k_0, k_1, k_2)\in \Z^3}N^{(m)}(k_0, k_1, k_2)t_0^{k_0}t_1^{k_1}t_2^{k_2} \nonumber
\\ \qquad{}
=
\begin{cases}
\!\!\displaystyle\binom{n}{m} \frac{t_2^{m-n}}{t_1^m}\frac{t_0(t_0t_1t_2)^{(2n^2+4n)/3} q_{n-1}(z)}{(z(z^2-1))^{(n^2-n)/3}}, &n\equiv 0, 1\!\!\! \mod 3,\!\!\!\vspace{.5ex}
\\
\!\!\displaystyle\binom{n}{m} \frac{t_2^{m-n}}{t_1^m}\frac{(t_0t_1+t_0t_2+t_1t_2)(t_0t_1t_2)^{(2n^2+4n-1)/3}
q_{n-1}(z)}{(3z^2+1)(z(z^2-1))^{(n^2-n-2)/3}},
&n\equiv 2\!\!\! \mod 3.
\end{cases}
\end{gather}
Observe that $k_0+k_1+k_2=(n+1)(2n+1)=2n^2+3n+1$ in each state.

\begin{Corollary}\label{numberofstateswith3colors}Let $N^{(m)}(k_0, k_1, k_2)$ be the number of states with $m$ positive turns and $k_i$ faces of color $i$. Then
\begin{gather*}
N^{(m)}(k_0, k_1, k_2)=\binom{n}{m}N^{(0)}(k_0, k_1+m, k_2-m).
\end{gather*}
\end{Corollary}

\begin{proof}
We can write \eqref{generatingfunction} as
\begin{gather}
\label{generalm}
\sum_{(k_0, k_1, k_2)\in \Z^3}N^{(m)}(k_0, k_1, k_2)t_0^{k_0}t_1^{k_1+m}t_2^{k_2-m}=\binom{n}{m}A,
\end{gather}
where $A$ does not depend on $m$.
For $m=0$, we get
\begin{gather}
\label{mequals0}
\sum_{(k_0, k_1, k_2)\in \Z^3}N^{(0)}(k_0, k_1, k_2)t_0^{k_0}t_1^{k_1}t_2^{k_2}
=A.
\end{gather}
Putting \eqref{mequals0} into \eqref{generalm} yields
\begin{gather*}
\sum_{(k_0, k_1, k_2)\in \Z^3}N^{(m)}(k_0, k_1, k_2)t_0^{k_0}t_1^{k_1+m}t_2^{k_2-m}=\binom{n}{m}\sum_{(k_0, k_1, k_2)\in \Z^3}N^{(0)}(k_0, k_1, k_2)t_0^{k_0}t_1^{k_1}t_2^{k_2}.
\end{gather*}
Since both sides are polynomials in $t_0$, $t_1$ and $t_2$, we can compare the coefficients of terms of the same multidegree pairwise and conclude that
\begin{gather*}
N^{(m)}(k_0, k_1, k_2)=\binom{n}{m}N^{(0)}(k_0, k_1+m, k_2-m). \tag*{\qed}
\end{gather*}
\renewcommand{\qed}{}
\end{proof}

Because of the above corollary, we only need to study the partition function for $m=0$.
Inspecting \eqref{generatingfunction}, one realizes that most factors are symmetric in $t_0$, $t_1$ and $t_2$. Because of these symmetries, a property for one color immediately implies a similar property for the other two colors.

\begin{Corollary}
\label{symmetriesofcolors}
Let $N^{(m)}(k_0, k_1, k_2)$ be the number of states with $m$ positive turns and $k_i$ faces of color $i$. Then
\begin{gather*}
N^{(m)}(k_0+d, k_1-m, k_2+m-n),
\end{gather*}
with
\begin{gather*}
d=
\begin{cases}
1, &\text{for}\quad n\equiv 0, 1 \mod 3,\\
0, & \text{for}\quad n\equiv 2 \mod 3,
\end{cases}
\end{gather*}
is a symmetric function of~$k_0$, $k_1$, $k_2$.
\end{Corollary}

\begin{proof}
We write
\eqref{generatingfunction} as
\begin{gather*}
\sum_{(k_0, k_1, k_2)\in \Z^3}N^{(m)}(k_0, k_1, k_2)t_0^{k_0}t_1^{k_1}t_2^{k_2}
=
\frac{t_2^{m-n}t_0^d}{t_1^m}S(t_0, t_1, t_2),
\end{gather*}
where $S(t_0, t_1, t_2)$ is symmetric in $t_0$, $t_1$ and $t_2$, and where $d=1$ for $n\equiv 0,1 \mod 3$, and $d=0$ for $n\equiv 2 \mod 3$.
Rearranging yields
\begin{gather*}
S(t_0, t_1, t_2)=\sum_{(k_0, k_1, k_2)\in \Z^3}N^{(m)}(k_0+d, k_1-m, k_2+m-n)t_0^{k_0}t_1^{k_1}t_2^{k_2}.
\end{gather*}
Since $S$ is symmetric in $t_0$, $t_1$ and $t_2$, it follows that $N^{(m)}(k_0+d, k_1-m, k_2+m-n)$ is symmetric in $k_0$, $k_1$ and $k_2$.
\end{proof}

Having a general formula for the partition function of the three-color model, we can also compute the minimum and maximum possible number of faces of each color. Since the exponents must be positive, we can read off the bounds in the coefficients, e.g., for color $0$, the bounds can be read off in \eqref{b}. The following corollary shows that the bounds are strict, and we find the number of states that reach the bounds.

\begin{Corollary}\label{maxminfaces}
Let $N^{(m)}_i(k)$ be the number of states with $m$ positive turns and $k$ faces of color~$i$. For each $m$, the number of states with the minimum number of faces of each color respectively is
\begin{gather*}
N^{(m)}_0\left(\frac{n^2+5n+a}{3}\right)=N^{(m)}_1\left(\frac{n^2+5n+c}{3}-m\right)
=N^{(m)}_2\left(\frac{n^2+2n+c}{3}+m\right)
=\binom{n}{m},
\end{gather*}
and the number of states with the maximum number of faces of each color is
\begin{gather*}
\begin{split}
N^{(m)}_0\left(\frac{5n^2+7n+2a}{6}\right)&=N^{(m)}_1\left(\frac{5n^2+7n+2c}{6}-m\right)
=N^{(m)}_2\left(\frac{5n^2+n+2c}{6}+m\right)
\\
&=\binom{n}{m}2^{n(n-1)/2},
\end{split}
\end{gather*}
where
\begin{gather*}
a=\begin{cases}
3, &\text{for}\quad n\equiv 0, 1 \mod 3,\\
1, &\text{for}\quad n\equiv 2 \mod 3,
\end{cases} \qquad \text{and}\qquad
c=\begin{cases}
0, &\text{for}\quad n\equiv 0, 1 \mod 3,\\
1, &\text{for}\quad n\equiv 2 \mod 3.
\end{cases}
\end{gather*}
\end{Corollary}
\begin{proof}
First consider the minimum number of faces of color $0$, for $m=0$. From the definition, we have $q_{n-1}(0)=1$. Computing $q_{n-1}(0)$ using \eqref{b} yields
\begin{gather*}
\begin{split}
q_{n-1}(z)\big|_{z=0}
&=\sum_{k_0}N^{(0)}_0(k_0) (z(z+1))^{k_0-(n^2+5n+a)/3} (z-1)^{(5n^2+7n+2a)/3-2k_0}\big|_{z=0}
\\
&=N^{(0)}_0\left(\frac{n^2+5n+a}{3}\right)\!.
\end{split}
\end{gather*}
Hence, by symmetry (Corollary~\ref{symmetriesofcolors}),
\begin{gather*}
N^{(0)}_2\left(\frac{n^2+2n+c}{3}\right)=N^{(0)}_1\left(\frac{n^2+5n+c}{3}\right)=N^{(0)}_0\left(\frac{n^2+5n+a}{3}\right)=1.
\end{gather*}

Now consider the maximum number of faces of color $0$, for $m=0$. In the limit $\psi=0$, we have $z=1$. Then \eqref{qnequalst2n} yields that
\begin{gather*}
q_{n-1}(z)\big|_{z=1}=T(2\psi+1, \dots, 2\psi+1)\big|_{\psi=0}.
\end{gather*}
From Section~4 in~\cite{Rosengren2014-1} we get that
\begin{gather*}
T(2\psi+1, \dots, 2\psi+1)\big|_{\psi=0}=2^{n(n-1)}.
\end{gather*}
On the other hand, computing $q_{n-1}(1)$ using \eqref{b} yields
\begin{gather*}
\begin{split}
q_{n-1}(z)\big|_{z=1}
&=\sum_{k_0}N^{(0)}_0(k_0) (z(z+1))^{k_0-(n^2+5n+a)/3} (z-1)^{(5n^2+7n+2a)/3-2k_0}\big|_{z=1}\\
&=N^{(0)}_0\left(\frac{5n^2+7n+2a}{6}\right) 2^{n(n-1)/2}.
\end{split}
\end{gather*}
Hence symmetry yields
\begin{gather*}
N^{(0)}_2\left(\frac{5n^2+n+2c}{6}\right)=N^{(0)}_1\left(\frac{5n^2+7n+2c}{6}\right)=N^{(0)}_0\left(\frac{5n^2+7n+2a}{6}\right)
=2^{n(n-1)/2}.
\end{gather*}

Corollary~\ref{numberofstateswith3colors} yields the desired results.
\end{proof}

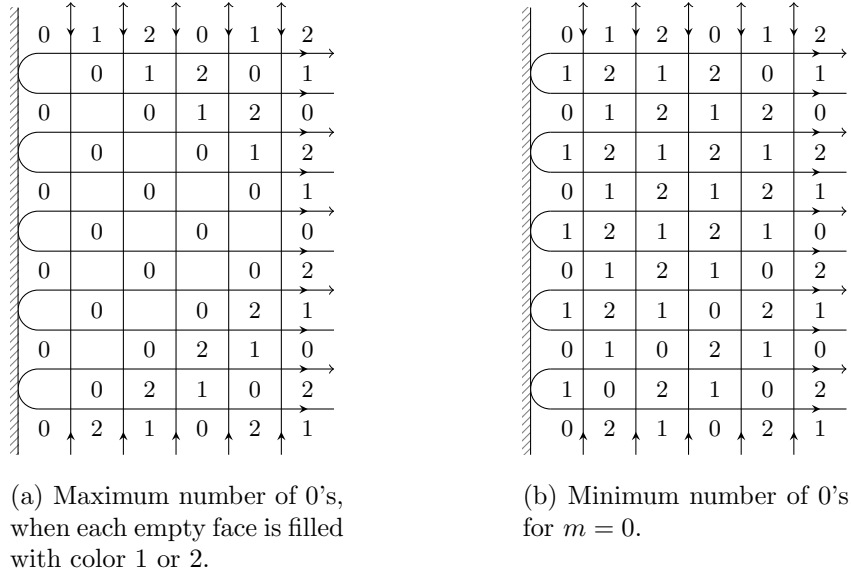
\begin{figure}[t]
\captionsetup[subfloat]{captionskip=10pt}
\centering
\subfloat[Maximum number of $0$'s, when each empty face is filled with color $1$ or $2$. \label{fig:maximum0}]{%
\begin{tikzpicture}[scale=0.7, font=\footnotesize]
	\foreach \y in {1,...,5} {
		\draw (.38,1.5*\y-.25-.38) -- +(4.3+0.32,0);
		\draw[midarrow={stealth}] (5,1.5*\y-.25-.38) -- +(1,0);
		\draw (.38,1.5*\y-.25+.38) -- +(4.3+0.32,0);
		\draw[midarrow={stealth},->] (5,1.5*\y-.25+.38) -- +(1,0);
		\draw (0.38,1.5*\y-.25+.38) arc (90:270:0.38);	
	}
	
	\fill[preaction={fill,white},pattern=north east lines, pattern color=gray] (0,0) rectangle (-.15,8.5) ; \draw (0,0) -- (0,8.5);
	
	\foreach \x in {1,...,5} {
		\draw[midarrow={stealth}] (\x,0) -- +(0,.87); 
		\draw (\x,.87) -- +(0,6.76); 
		
		\draw[midarrow={stealth reversed}, ->] (\x,7.63) -- +(0,.97);	
	}

		\node at (0.5, 8) {$0$};
 \node at (1.5, 8) {$1$};
 \node at (2.5, 8) {$2$};
 \node at (3.5, 8) {$0$};
 \node at (4.5, 8) {$1$};
 \node at (5.5, 8) {$2$};

		\node at (5.5, 7.25) {$1$};
		\node at (5.5, 6.5) {$0$};
 \node at (5.5, 5.75) {$2$};
		\node at (5.5, 5) {$1$};
		\node at (5.5, 4.25) {$0$};
 \node at (5.5, 3.5) {$2$};
 \node at (5.5, 2.75) {$1$};
 \node at (5.5, 2) {$0$};
 \node at (5.5, 1.25) {$2$};

 \node at (5.5, 0.5) {$1$};
		\node at (4.5, 0.5) {$2$};
		\node at (3.5, 0.5) {$0$};
 \node at (2.5, 0.5) {$1$};
 \node at (1.5, 0.5) {$2$};
 \node at (0.5, 0.5) {$0$};
		
		\node at (0.5, 2) {$0$};
		\node at (0.5, 3.5) {$0$};
		\node at (0.5, 5) {$0$};
		\node at (0.5, 6.5) {$0$};

		\node at (1.5, 7.25) {$0$};
		\node at (2.5, 6.5) {$0$};
 \node at (3.5, 5.75) {$0$};
		\node at (4.5, 5) {$0$};
		
		\node at (1.5, 5.75) {$0$};
		\node at (2.5, 5) {$0$};
 \node at (3.5, 4.25) {$0$};
		\node at (4.5, 3.5) {$0$};
		
		\node at (1.5, 1.25) {$0$};
		\node at (2.5, 2) {$0$};
 \node at (3.5, 2.75) {$0$};

		\node at (1.5, 2.75) {$0$};
		\node at (1.5, 4.25) {$0$};
 \node at (2.5, 3.5) {$0$};
		
		\node at (2.5, 7.25) {$1$};
 \node at (3.5, 6.5) {$1$};
		\node at (4.5, 5.75) {$1$};
		
 \node at (3.5, 7.25) {$2$};
		\node at (4.5, 6.5) {$2$};
		
		\node at (4.5, 7.25) {$0$};
	
		\node at (2.5, 1.25) {$2$};
 \node at (3.5, 2) {$2$};
		\node at (4.5, 2.75) {$2$};
		
 \node at (3.5, 1.25) {$1$};
		\node at (4.5, 2) {$1$};
		
		\node at (4.5, 1.25) {$0$};
\end{tikzpicture}
	}\hfil
	\subfloat[Minimum number of $0$'s, for $m=0$. \label{fig:minimum0}]{%
\begin{tikzpicture}[scale=0.7, font=\footnotesize]
	\foreach \y in {1,...,5} {
		\draw (.38,1.5*\y-.25-.38) -- +(4.3+0.32,0);
		\draw[midarrow={stealth}] (5,1.5*\y-.25-.38) -- +(1,0);
		\draw (.38,1.5*\y-.25+.38) -- +(4.3+0.32,0);
		\draw[midarrow={stealth},->] (5,1.5*\y-.25+.38) -- +(1,0);
		\draw (0.38,1.5*\y-.25+.38) arc (90:270:0.38);
	}
	
	\fill[preaction={fill,white},pattern=north east lines, pattern color=gray] (0,0) rectangle (-.15,8.5) ; \draw (0,0) -- (0,8.5);
	
	\foreach \x in {1,...,5} {
		\draw[midarrow={stealth}] (\x,0) -- +(0,.87); 
		\draw (\x,.87) -- +(0,6.76); 
		
		\draw[midarrow={stealth reversed}, ->] (\x,7.63) -- +(0,.97);	
	}

		\node at (0.7, 8) {$0$};
 \node at (1.5, 8) {$1$};
 \node at (2.5, 8) {$2$};
 \node at (3.5, 8) {$0$};
 \node at (4.5, 8) {$1$};
 \node at (5.5, 8) {$2$};

		\node at (5.5, 7.25) {$1$};
		\node at (5.5, 6.5) {$0$};
 \node at (5.5, 5.75) {$2$};
		\node at (5.5, 5) {$1$};
		\node at (5.5, 4.25) {$0$};
 \node at (5.5, 3.5) {$2$};
 \node at (5.5, 2.75) {$1$};
 \node at (5.5, 2) {$0$};
 \node at (5.5, 1.25) {$2$};

 \node at (5.5, 0.5) {$1$};
		\node at (4.5, 0.5) {$2$};
		\node at (3.5, 0.5) {$0$};
 \node at (2.5, 0.5) {$1$};
 \node at (1.5, 0.5) {$2$};
 \node at (0.7, 0.5) {$0$};
		
		\node at (0.7, 2) {$0$};
		\node at (0.7, 3.5) {$0$};
		\node at (0.7, 5) {$0$};
		\node at (0.7, 6.5) {$0$};

		\node at (1.5, 1.25) {$0$};
		\node at (2.5, 2) {$0$};
 \node at (3.5, 2.75) {$0$};
		\node at (4.5, 3.5) {$0$};
		
		\node at(0.7, 1.25) {$1$};
		\node at (1.5, 2) {$1$};
		\node at (2.5, 2.75) {$1$};
 \node at (3.5, 3.5) {$1$};
		\node at (4.5, 4.25) {$1$};			
		
		\node at (1.5, 2.75) {$2$};
		\node at (2.5, 3.5) {$2$};
 \node at (3.5, 4.25) {$2$};
		\node at (4.5, 5) {$2$};	
		
		\node at(0.7, 2.75) {$1$};
		\node at (1.5, 3.5) {$1$};
		\node at (2.5, 4.25) {$1$};
 \node at (3.5, 5) {$1$};
		
		\node at (1.5, 4.25) {$2$};
		\node at (2.5, 5) {$2$};
 \node at (3.5, 5.75) {$2$};
		
		\node at(0.7, 4.25) {$1$};
		\node at (1.5, 5) {$1$};
		\node at (2.5, 5.75) {$1$};
		
		\node at (1.5, 5.75) {$2$};
		\node at (2.5, 6.5) {$2$};
		
		\node at(0.7, 5.75) {$1$};
		\node at (1.5, 6.5) {$1$};
		
		\node at (1.5, 7.25) {$2$};
		\node at (0.7, 7.25) {$1$};
		
		\node at (2.5, 7.25) {$1$};
 \node at (3.5, 6.5) {$1$};
		\node at (4.5, 5.75) {$1$};
		
 \node at (3.5, 7.25) {$2$};
		\node at (4.5, 6.5) {$2$};
		
		\node at (4.5, 7.25) {$0$};
	
		\node at (2.5, 1.25) {$2$};
 \node at (3.5, 2) {$2$};
		\node at (4.5, 2.75) {$2$};
		
 \node at (3.5, 1.25) {$1$};
		\node at (4.5, 2) {$1$};
		
		\node at (4.5, 1.25) {$0$};
\end{tikzpicture}
}\\
\caption{States with the maximum and minimum number of faces with color $0$, for $n=5$.}
\label{fig:worstcases}
\end{figure}

Observe that the minimum and maximum number of faces of color 1 and 2 depend on the number of positive turns, $m$, whereas for color $0$, the minimum and maximum numbers respectively are the same for all $m$.

Some of the results in the theorem above, we can find combinatorially. The maximum of~$k_0$ corresponds to states looking as in Fig.~\ref{fig:maximum0}. There is a triangle in the middle where every second face has a $0$, and every other second face can be filled with either color $1$ or color $2$. For each given configuration of the turns, there is a total of~$\binom{n}{2}$ faces with a choice, which yields $2^{\binom{n}{2}}$ states obtaining the maximum of~$k_0$, which is in line with the result in Corollary~\ref{maxminfaces}.
Considering the states with $0$ positive turns, the minimum of~$k_0$ is depicted in Fig.~\ref{fig:minimum0}. Here only the lower diagonal border of the triangle has zeroes, whereas in the middle, we have a chess board pattern of color 1 and 2. There is only one such state for $m=0$.
For both the maximums and the minimums, the upper and lower right corner triangles can be filled up in such a way that each third diagonal consists of the same color.

For a fixed configuration of the turns, the empty faces on the left boundary in Fig.~\ref{fig:maximum0} are fixed. The minimum of~$k_1$ is the state where the remaining empty faces are filled up with a $2$. Similarly for the minimum of~$k_2$, the empty faces are to be filled up with a $1$.
Considering the states with $m=0$, we can find the maximum of~$k_1$ by putting a $1$ in all the empty faces in Fig.~\ref{fig:maximum0}. Another maximum is in Fig.~\ref{fig:minimum0}. All the states with maximum of~$k_1$ will have the $1$'s in the same place. The faces that can differ are the $\binom{n}{2}$ faces that have a $0$ in Fig.~\ref{fig:maximum0}, and a $2$ in Fig.~\ref{fig:minimum0}. All these faces can have either a $0$ or a $2$, which results in $2^{\binom{n}{2}}$ states with the maximum of~$k_1$ for $m=0$, which is in line with the result in Corollary~\ref{maxminfaces}. For $m=n$, we can find all the states with the maximum of~$k_2$ in a similar way.

For a general $m$, there are $\binom{n}{m}$ ways to choose the turns that should be positive, which then yields the total number of states with the minimum of~$k_1$ and $k_2$, and the maximum of~$k_0$ respectively. All these states are variations of Fig.~\ref{fig:maximum0}. It~is easy to find all these states combinatorially, since the $0$'s are fixed and are not affected by the number in the
turns. It~seems harder to explicitly find all the maximums of~$k_1$ and $k_2$ and all the minimums of~$k_0$ combinatorially for a general $m$, since a change on the face in a turn could force a change on the face beside it, which in turn could force more changes. Nevertheless, algebraically we can find the number of states with the maximum or minimum of~$k_i$, for all colors $i$, using the above corollary.

Kuperberg \cite{Kuperberg2002} stated a formula for counting the number of UASMs, which is equivalent to the number of states in the 8VSOS model with DWBC and reflecting end and in the corresponding three-color model. As a corollary of Theorem~\ref{maintheorem} we can find this number.

\begin{Corollary}[Kuperberg]
\label{numberofstates}
For a fixed $n$, the number of states with $m$ positive turns in the 8VSOS model with DWBC and reflecting end is
\begin{gather*}
A^m_n=\binom{n}{m}\frac{1}{2^n} \prod_{i=0}^{n-1}\frac{(2i+1)!(6i+4)!}{(4i+2)!(4i+3)!}.
\end{gather*}
\end{Corollary}

\begin{proof}
In Theorem~\ref{maintheorem}, put $t_0=t_1=t_2=1$. Then $z=1/3$. Consider $m=0$. Then the sum on the left hand side in Theorem~\ref{maintheorem} counts the number of states with $0$ positive turns.
For~all~$n$, the formula becomes
\begin{gather*}
A^0_n=(3/2)^{n^2-n}q_{n-1}(1/3).
\end{gather*}
Because of the symmetries \eqref{symmetries},
\begin{gather*}
q_{n-1}(1/3)=(2/3)^{n(n-1)} x^{n(n-1)}q_{n-1}(1/x)\big|_{x=0}.
\end{gather*}
Hence for all $n$,
\begin{gather*}
A^0_n=x^{n(n-1)}q_{n-1}(1/x)\big|_{x=0}.
\end{gather*}
{\sloppy
Since $\deg q_{n-1}(z)=n(n-1)$, the number $x^{n(n-1)}q_{n-1}(1/x)\big|_{x=0}$ is the leading coefficient of~$q_{n-1}(z)$. In \cite{BazhanovMangazeev2010} a formula for these numbers is given:
\begin{gather*}
x^{n(n-1)}q_{n-1}\left(\frac{1}{x}\right)\bigg|_{x=0}=\frac{1}{2^n}\prod_{i=0}^{n-1}\frac{(2i+1)!(6i+4)!}{(4i+2)!(4i+3)!}.
\end{gather*}

}

Varying $m$ in the formula in Theorem~\ref{maintheorem}, the only thing that is affected is the binomial coefficient.
Hence
$A^m_n=\binom{n}{m} A^0_n.$
This yields the desired result.
\end{proof}

In the above corollary, $m=0$ corresponds to the number of VSASMs.

\subsection*{Acknowledgements} I would like to thank my supervisor Hjalmar Rosengren and my co-supervisor Jules Lamers for~the numerous hours of support you have given me throughout the whole research process and~while writing this article. I also would like to thank the anonymous referees for many useful comments and suggestions.

\pdfbookmark[1]{References}{ref}
\LastPageEnding

\end{document}